\documentclass[journal,comsoc,10pt]{IEEEtran}

\usepackage{xcolor}
\usepackage{tikz}
\usetikzlibrary{er,automata,positioning,chains,fit,shapes,calc,arrows,plotmarks,arrows.meta,shapes}

\usepackage{caption}
\usepackage{booktabs}
\usepackage{cite}

\usepackage{pgfplots}
\usepackage{pgfplotstable}

\usepackage{amsmath,amssymb,amsfonts}
\usepackage{amsthm,bm}  
\usepackage{bbm}
\usepackage{graphicx}
\usepackage{textcomp}
\usepackage{dsfont}
\usepackage{mathtools}

\usetikzlibrary{external}

\newtheorem{theorem}{Theorem}

\newtheorem{corollary}{Corollary}

\newtheorem{innercustomgeneric}{\customgenericname}
\providecommand{\customgenericname}{}
\newcommand{\newcustomtheorem}[2]{%
  \newenvironment{#1}[1]
  {%
   \renewcommand\customgenericname{#2}%
   \renewcommand\theinnercustomgeneric{##1}%
   \innercustomgeneric
  }
  {\endinnercustomgeneric}
}

\newcustomtheorem{customthm}{Theorem}
\newcustomtheorem{customlemma}{Lemma}
\newcustomtheorem{customprop}{Proposition}
\newcustomtheorem{customcor}{Corollary}

\theoremstyle{definition}
\newtheorem{definition}{Definition}
 \newtheorem{example}{Example}
  \newtheorem*{example*}{Example}
 \newtheorem{remark}{Remark}
 \newtheorem*{remark*}{Remark}

\newcommand{\Ie}{\textit{i.e., }}
\newcommand{\Eg}{\textit{e.g., }}

\newcommand{\cA}{\mathcal{A}}

\newcommand{\cG}{\mathcal{G}}

\newcommand{\cM}{\mathcal{M}}

\newcommand{\cO}{\mathcal{O}}

\newcommand{\cQ}{\mathcal{Q}}

\newcommand{\cV}{\mathcal{V}}

\newcommand{\cX}{\mathcal{X}}

\usepackage{mathrsfs}

\newcommand{\bS}{\mathbf{S}}

\newcommand{\ep}{\epsilon}
\newcommand{\Er}{\mathcal{E}}

\newcommand{\hbS}{\hat{\mathbf{S}}}

\newcommand{\E}{\mathbbm{E}}

\newcommand{\QA}{\textnormal{QA}}
\newcommand{\RG}{\textnormal{RG}}

\newcommand{\MSE}{\textnormal{MSE}}

\newcommand{\utag}[2]{\mathop{#2}\limits^{\text{(#1)}}}
\newcommand{\uref}[1]{(#1)}

\definecolor{DarkGreen}{rgb}{0.1,0.5,0.1}
\definecolor{DarkRed}{rgb}{0.5,0.1,0.1}
\definecolor{DarkBlue}{rgb}{0.1,0.1,0.5}
\definecolor{DarkPurple}{rgb}{0.5,0.2,0.5}
\definecolor{DarkTurquoise}{rgb}{0.1,0.5,0.5}
\definecolor{darkorchid}{rgb}{0.6, 0.2, 0.8}
\definecolor{americanrose}{rgb}{1.0, 0.01, 0.24}
\definecolor{darkpastelgreen}{rgb}{0.01, 0.75, 0.24}
\definecolor{bleudefrance}{RGB}{97, 142, 206}
\definecolor{brightlavender}{RGB}{223, 196, 236}
\definecolor{palegreen}{RGB}{200, 240, 208}
\definecolor{dimgray}{rgb}{0.41, 0.41, 0.41}
\definecolor{fireenginered}{rgb}{0.81, 0.09, 0.13}
\definecolor{denim}{rgb}{0.08, 0.38, 0.74}
\definecolor{verylightgray}{rgb}{0.75, 0.75, 0.75}
\definecolor{lavendergray}{RGB}{218,232,229}
\definecolor{verylightgreen}{RGB}{247, 254, 247}
\definecolor{verylightblue}{RGB}{249, 251, 254}
\definecolor{verylightred}{RGB}{254, 247, 247}
\definecolor{grannysmithapple}{RGB}{178, 249, 176}
\definecolor{lavenderpurple}{RGB}{242, 183, 234}

\usepackage{subcaption}
\usepackage{cases}

\usepackage{tablefootnote}
\usepackage{array, makecell}
\usepackage{multirow}

\allowdisplaybreaks

\usepackage[linesnumbered,ruled,vlined]{algorithm2e}
\usepackage{algpseudocode}

\SetKwInput{KwInput}{Input}                
\SetKwInput{KwOutput}{Output}              

\algnewcommand\algorithmicforeach{\textbf{for each}}
\algdef{S}[FOR]{ForEach}[1]{\algorithmicforeach\ #1\ \algorithmicdo}

\usepackage{balance}

\usepackage{enumitem}

\usepackage[cmintegrals]{newtxmath}

\begin{document}

\title{Private Multi-Group Aggregation
\\ \thanks{This work was  presented
in part at the IEEE International Symposium on Information Theory 2021.}
}

\author{\IEEEauthorblockN{Carolina Naim\IEEEauthorrefmark{1}, Rafael G. L. D'Oliveira\IEEEauthorrefmark{2}, and Salim El Rouayheb\IEEEauthorrefmark{1}\\}
\IEEEauthorblockA{\IEEEauthorrefmark{1}ECE, Rutgers University, \{carolina.naim, salim.elrouayheb\}@rutgers.edu\\ 
\IEEEauthorrefmark{2}RLE, Massachusetts Institute of Technology, rafaeld@mit.edu \\
}
}

\maketitle
\pagenumbering{gobble}
\begin{abstract}
We study the differentially private multi group aggregation (PMGA) problem. This setting involves a single server and $n$ users. Each user belongs to one of $k$ distinct groups and holds a discrete value. The goal is to design schemes that allow the server to find the aggregate (sum) of the values in each group (with high accuracy) under communication and local differential privacy constraints. The privacy constraint guarantees that the user's group remains private. This is motivated by applications where a user's group can reveal sensitive information, such as his religious and political beliefs, health condition, or race.  

We propose a novel scheme, dubbed Query and Aggregate (Q\&A) for PMGA. The novelty of Q\&A is that it is an interactive aggregation scheme. In Q\&A,  each user is assigned a random query matrix, to which he sends the server an answer based on his group and value. We characterize the Q\&A scheme's performance in terms of accuracy (MSE), privacy, and communication. We compare Q\&A to the Randomized Group (RG) scheme, which is non-interactive and adapts existing randomized response schemes to the PMGA setting. We observe that typically Q\&A outperforms RG, in terms of  privacy vs. utility, in the high privacy regime. 
\end{abstract}

\begin{IEEEkeywords}
Differential privacy, data privacy, estimation.
\end{IEEEkeywords}

\section{Introduction}

We consider the problem of distributed aggregation in which a centralized server wishes to compute the aggregate (sum) of the data (values) held by several users. 
Privacy is a significant concern since   participants have to share their data, which can be personal and sensitive. This has motivated works on private and secure distributed aggregation  in many applications such as medical studies \cite{Kim_2014} or, more recently, federated learning \cite{bonawitz_2016, Abadi_2016,Truex_2020, Kim_2021}.

In this work, we focus on the setting depicted in Figure~\ref{fig:setting}, in which users belong to different groups. The server wants to find the aggregate for each group separately. As opposed to finding the aggregate over the whole population, as is typical distributed aggregation problems,  \Eg \cite{bonawitz_2016, Chan_2012}. The users' groups can be based, for example, on their political views, immigration status, health condition, or race,  to name a few. This raises additional privacy concerns since participating users may be rightfully wary of revealing their group.

Consider, for example, medical or clinical trials conducted to determine how having a certain illness, say diabetes, affects the    efficacy of a new vaccine.  A volunteer may want to contribute his vaccine test results, but does not want to reveal his medical condition (diabetes), \Ie the group he belongs to.
Another application is population polling during elections, where pollsters want to estimate how different groups of the population vote. Such groups could depend on race, gender, age, or income bracket. 
The poll participants want to indicate which political candidate they will vote for while keeping their group private.


 Motivated by these examples, we present the problem of Private Multi-Group Aggregation (PMGA), where local differential privacy \cite{Dwork_2006,Duchi_2013} guarantees are given over a user's group. We are interested in schemes that scale well with the number of groups since more groups allow the server more refined statistics about the population.  Our main objective is to design schemes with low communication costs per user, as users can have limited bandwidth. In particular, we focus on schemes that offer communication costs that are constant or at most logarithmic in the number of groups.
 Moreover, we study the trade-offs they offer between privacy (measured using local differential privacy) and accuracy, \Ie the aggregate estimator's mean square error.
 

\begin{figure}[t]
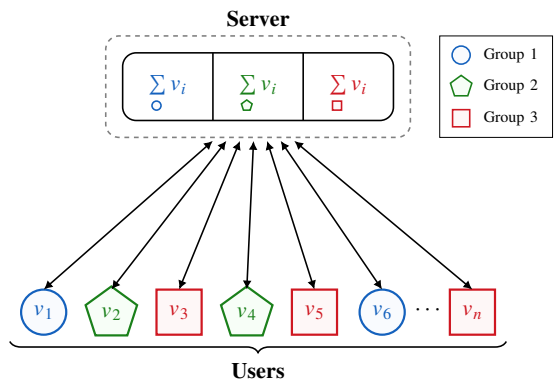

    \centering
    \scalebox{0.6}{

    }
    \caption{An instance of the Private Multi-Group Aggregation problem with    $n$ users. Each user $i$,  for $i\in\{1,\dots,n\}$, has a scalar value, $v_i$, and belongs to one of the $k=3$ distinct groups. The server's goal is to estimate the sum of the values in each of the  groups. }
    \label{fig:setting}
\end{figure}
\begin{figure*}[t]
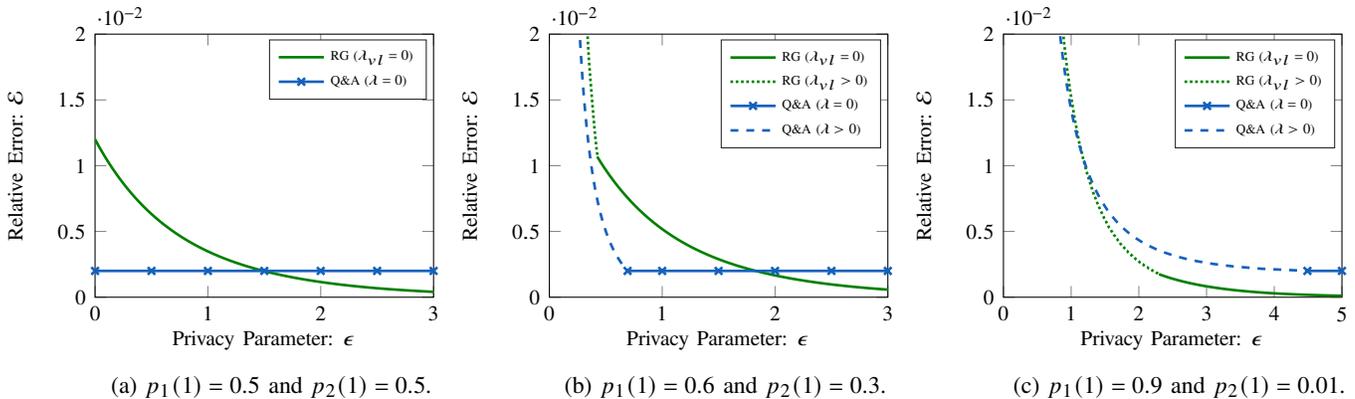

\hspace{9pt}
     \begin{subfigure}[c]{0.31\textwidth}
         \centering
             \definecolor{darkolivegreen}{rgb}{0.0, 0.5, 0.0}
\definecolor{denim}{rgb}{0.08, 0.38, 0.74}
%
%
         \caption{\centering  $p_1(1)=0.5$ and $p_2(1)=0.5$.}
         \label{fig:MSE_vs_ep1}
     \end{subfigure}
     \hfill
     \begin{subfigure}[c]{0.34\textwidth}
         \centering
         \definecolor{darkolivegreen}{rgb}{0.0, 0.5, 0.0}
\definecolor{denim}{rgb}{0.08, 0.38, 0.74}
%
%
         \caption{\centering  $p_1(1)=0.6$ and $p_2(1)=0.3$. }
         \label{fig:MSE_vs_ep2}
     \end{subfigure}
     \hfill
      \begin{subfigure}[c]{0.31\textwidth}
         \centering
         \definecolor{darkolivegreen}{rgb}{0.0, 0.5, 0.0}
\definecolor{denim}{rgb}{0.08, 0.38, 0.74}
%
%
         \caption{\centering  $p_1(1)=0.9$ and $p_2(1)=0.01$. }
         \label{fig:MSE_vs_ep3}
     \end{subfigure}
     
        \caption{Comparison of the Q\&A and RG schemes for $k=2$ groups, binary alphabet, \Ie $v\in\{-1,1\}$, and fixed total communication cost equal to $500$ bits, \Ie $500$ bits communicated by all the users to the server. The Q\&A scheme requires less communication cost per user compared to the RG scheme; therefore, for fixed total communication cost, the Q\&A scheme has more users.
        The subfigures (a), (b) and (c) illustrate accuracy vs. privacy of the Q\&A and RG schemes for different user's value distributions, $p_1(1)$ and $p_2(1)$. The dashed (or dotted) curves represent the performance of the schemes with an additional layer of privacy that hides the user's values, \Ie $\lambda>0$ for the Q\&A scheme, and $\lambda_{vl}>0$ for the RG scheme. We present a more detailed comparison is in Section \ref{sec:comparisons}.}
        \label{fig:comparisons}
\end{figure*}

\subsection{Related Work}

The classical setup for secure and private aggregation in the literature does not distinguish among groups, and the privacy guarantees are on the users' data (values).
Differentially private schemes and bounds for private aggregation were studied in \cite{Chan_2012,Ghazi_2020,Goryczka_2017,Truex_2019,Shi_2011}.
  In \cite{bonawitz_2016}, secure aggregation based on information-theoretic (secret sharing) and cryptographic techniques was developed for applications to federated learning (FL) \cite{McMahan_2017}. Secure aggregation algorithms for FL with improved communication and computation overhead were proposed in \cite{So_2020,kadhe_2020}, and with robustness against adversarial users in \cite{Pillutla_2019}. 
{These  schemes have a  per-user communication cost  that grows with the number of  users.}

{Although in this paper we focus on estimating the sum, other works have focused on various estimation problems.}
For instance, distributed empirical mean estimation under communication constraints has been looked at in \cite{Konecny_2018,suresh_2017}.  Beyond estimating the mean, discrete distribution estimation under communication constraints has been studied in \cite{Barnes_2019,Diakonikolas_2017}, and under privacy constraints in \cite{Kairouz_2016,Ye_2017, Diakonikolas_2015, Acharya_2019, Wang_2016, Erlingsson_2014}. Moreover, heavy hitters  (most frequent items)  estimation has also been studied in \cite{Acharya_2019_2,Zhu_2019} under privacy and communication constraints. Recent work in \cite{Chen_2020} devises schemes that achieve optimal privacy and communication for  mean  and frequency estimation. 

Another related problem is federated submodel learning \cite{Niu_2019,kim_2020,Jia_2020}. In this setting, one or  multiple servers  hold  various submodels (vectors) and each user wants to train (update) a private subset of these submodels. {The notion of submodels here is similar to the notion of groups in our problem; however, a user's update depends on the submodels at the server in addition to his data.} The proposed solutions in \cite{kim_2020,Jia_2020} use information-theoretic private information retrieval (PIR) to privately download and update the submodels. Thus, they require multiple servers, and the communication  cost per user is linear in the number of submodels (groups). 
Moreover, in \cite{Niu_2019}  differentially private techniques  were used to allow a user to download the required submodels, and update it using secure aggregation.
The resulting scheme has a communication cost per user that grows with the total number of  users.

\subsection{Contributions}
 {We introduce the problem of private multi-group aggregation (Figure \ref{fig:setting}), where $n$ users communicate with a central server. Each user holds a value and belongs to a private group.
 The goal is for the server to accurately compute the sum of values per group while keeping the user's group private.
 The notion of privacy we use is local differential privacy.}


Our main contribution is a novel scheme for PMGA that we call the Query and Aggregate (Q\&A) scheme that provides local differential privacy guarantees on the users' groups. The Q\&A scheme is interactive in that
the user is assigned a query matrix and sends the server an answer based on his group and value. This allows to shift the bulk of the total communication cost to the query stage (server-to-user) which can be done offline since it does not depend on a user's group and value. Thus, the online user-to-server communication cost does not depend on the number of groups and users, as typically occurs in secure aggregation problems, \Eg \cite{bonawitz_2016}. In Theorem \ref{thm:QA}, we characterize the performance of the Q\&A scheme in terms of privacy, communication cost, and accuracy.

We compare Q\&A to a non-interactive scheme which we call the Randomized Group (RG) scheme. RG is an adaptation of standard randomized response \cite{Warner_1965} schemes from the literature and consists of each user reporting a noisy version of his group and value to the server.
For a fixed total communication cost, we observe that in general Q\&A  offers better accuracy in high privacy regimes (small $\ep$), as illustrated in Figure~\ref{fig:comparisons}.

\subsection{Paper Organization}
The rest of the paper is organized as follows. In Section~\ref{sec:problem_formulation}, we describe the formulation of the Private Multi-Group Aggregation problem. In Section \ref{sec:main_results}, we present our main results, which consist of the Query and Aggregate (Q\&A) scheme and its performance (Theorem \ref{thm:QA}) compared to our proposed Randomized Group (RG) scheme. We present the details of the Q\&A scheme in Section~\ref{sec:QA_Scheme}, and those of the RG scheme in Section \ref{sec:Scheme_R}. We compare the two schemes in Section \ref{sec:comparisons}. Finally, we conclude and give future directions in Section \ref{sec:conclusion}.

\subsection{Notation}
We represent random variables by upper case letters, \Eg $X$, realizations of these random variables by lower case letters, \Eg $x$, and the alphabets of the random variables by calligraphic letters, \Eg $\cX$.
We use $\log(x)=\log_2(x)$ and $\ln(x)=\log_e(x)$. 
Also, for any positive integer $n$, we denote $[n]:=\{1,\dots,n\}$. 
Moreover, we use a colon to refer to whole rows or columns in a matrix or vector. For instance, $A(:,i)=(A(1,i),A(2,i),\dots,A(n,i))^\top$ is the $i^{th}$ column of $A$. 
{The $L_2$-norm of a vector $\mathbf{x}$ is denoted by $||\mathbf{x}||$.}


\section{Problem Formulation}
\label{sec:problem_formulation}
We consider the setting depicted in Figure~\ref{fig:setting}  in which there are  $n$ users, indexed from $1$ to $n$, and  a  single  server.  The users can communicate with the server but not among each other.  Each user $i\in[n]:=\{1,\dots,n\}$ belongs to one of $k$ groups, indexed from $1$ to $k$. Moreover, user $i\in[n]$ holds a value $v_i\in\cV:=\{\pm 1,\dots,\pm m\}$.
We assume that the server knows each user's index but does not know his value or group. We assume that the users are not adversarial and  faithfully participate in the scheme.

We denote by $G_i$ the random variable representing the group that user $i$ belongs to. We assume that the $G_i$, for all $i\in[n]$, are  identical and independent random variables from the alphabet $\cG=\{1,\dots,k\}$. The probability that any user $i\in[n]$ belongs to group $g\in\cG$ is denoted by $\theta_g:=\Pr(G_i=g).$
We denote by    $\theta:=(\theta_1,\dots,\theta_k)$ the realization of the random vector $\Theta$.

Each user $i$  in group $g$ holds an independent random scalar value $V_i$ drawn from the  alphabet  $\cV:=\{\pm 1,\dots,\pm m\}$ according to the  distribution ${p}_g(v):=\Pr(V_i=v|G_i=g)$.
The values of the users in the same group are independent and identically distributed.
We represent the users' value distributions by a $k\times 2m$ matrix
\begin{equation*}
    p:=\left[\begin{smallmatrix} p_1(-m) & \dots &p_1(m) \\
 \vdots& \ddots & \vdots\\
 p_k(-m)& \dots & p_k(m)\end{smallmatrix}\right].
\end{equation*}
The matrix $p$ is unknown, to both the users and the server, and is assumed to be the realization of a random variable $P$. 
Given their group $g\in\cG$, for all $v\in\cV$, the users behave identically, \Ie  ${p}_g(v)=\Pr(V_i=v|G_i=g)$ for any $i\in[n]$.

User $i$ knows the realizations of the random variables $G_i$ and $V_i$ representing his group and value. However, the distribution of the random variables $P$ and $\Theta$, and their realizations, are not necessarily known  neither to the server nor to the user. 

 The goal is to design a scheme that allows the server to compute an estimate of  the sum of values per group, \Ie to estimate the aggregate vector $\bS\in \mathbb{Z}^k$ with 
 \begin{equation*}
     \bS(g)=\sum_{i\in[n]:g_i= g} v_i,\quad \text{for all } g\in\cG.
 \end{equation*}  

 We consider schemes where each user $i$ can be assigned a query $q_i\in\cQ$, which is also known to the server. In response to the query, the user sends the server an answer $a_i\in\cA$. Upon receiving the answers from all users, the server finds an estimate $\hbS$ of $\bS$. 
 We characterize the efficiency of a scheme according to (i) communication, (ii) accuracy, and (iii) privacy. 
\begin{enumerate}[label=\textbf{(\roman*)},wide, labelindent=0pt]
    \item \textbf{\emph{Communication:}} We characterize the communication cost by the number of bits communicated between the server and the user. We look at the communication cost from two vantage points:  (i) user-centric, that measures the communication per user, \Ie the number of bits communicated between a user and the server; and (ii) server-centric,  that   measures the total communication the server receives from  all the users. We refer to the latter as the \textit{total communication cost}. 
    
    \item \textbf{\emph{Accuracy:}}  We use the relative mean square error to measure the accuracy of a scheme $\pi$. The risk of the estimator $\hbS_\pi$ is 
    \begin{equation}
    \label{eq:error_def}
        \Er_\pi:=\frac{1}{n^2}\MSE(\hbS_\pi),
    \end{equation}
    where $\MSE(\hbS_\pi)=\E\left[||\hbS_\pi-\bS||^2\big|P=p,\Theta=\theta\right]$. For ease of notation, the conditioning on $P$ and $\Theta$ is implicit in the rest of the paper. 
    The relative mean square error captures  the accuracy of our estimate relative to the expected true aggregate $\bS$. Since  $\E\left[||\bS||^2\right]$ grows as $\cO(n^2)$, we normalize by ${n^2}$.
   

    
    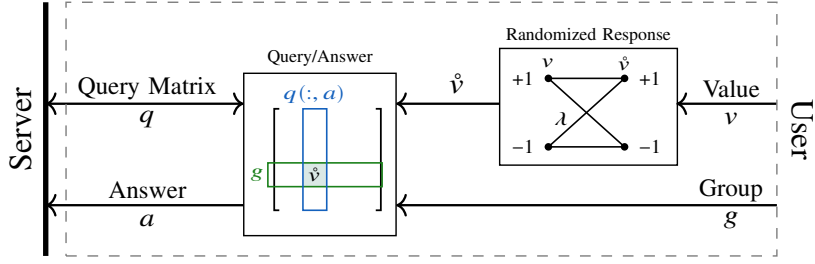
\begin{figure*}[t]
    \centering
    \scalebox{1.35}{
\begin{tikzpicture}
\definecolor{fireenginered}{rgb}{0.81, 0.09, 0.13}
\definecolor{denim}{rgb}{0.08, 0.38, 0.74}

\node[draw, rectangle, minimum width=1.75cm,minimum height=1.15cm] at (6.095,2.22) {};
\node at (6.095,2.92) {\tiny Randomized Response};
\node[fill,circle, inner sep=0.75pt] at (6.45,2.5) {};
\node[fill,circle, inner sep=0.75pt] at (6.45,1.825) {};
\node at (6.7,2.5) {\tiny$+1$};
\node at (6.7,1.825) {\tiny$-1$};
\node at (6.45,2.65) {\tiny$\mathring{v}$};
\node[fill,circle, inner sep=0.75pt] at (5.7,2.5) {};
\node[fill,circle, inner sep=0.75pt] at (5.7,1.825) {};
\node at (5.45,2.5) {\tiny$+1$};
\node at (5.45,1.825) {\tiny$-1$};
\node at (5.7,2.65) {\tiny$v$};
\draw[-,line width=0.45] (6.45,2.5)--(5.7,2.5);
\draw[-,line width=0.45] (6.45,2.5)--(5.7,1.825);
\draw[-,line width=0.45] (6.45,1.825)--(5.7,2.5);
\draw[-,line width=0.45] (6.45,1.825) --(5.7,1.825);
\node at (5.825,2.1) {\tiny$\lambda$};

\node[draw, rectangle, minimum width=1.5cm,minimum height=1.6cm] at (3.45,1.75) {};
\node at (3.45,2.7) {\tiny Query/Answer};

\draw[-,color=black, line width=0.5pt] (4,1.2) -- (4.05,1.2) -- (4.05,2.2) -- (4,2.2);
\draw[-,color=black, line width=0.5pt] (3.05,1.2) -- (3,1.2) -- (3,2.2) -- (3.05,2.2);

\node[fill, color=lavendergray, rectangle, minimum width=0.2cm, minimum height=0.2cm, line width=0.45pt] at (3.4,1.55) {};
\node[color=black] at (3.4,1.55) {\tiny$\mathring{v}$};

\node[draw, color=DarkGreen, rectangle, minimum width=1.125cm, minimum height=0.2cm, line width=0.45pt] at (3.5,1.55) {};
\node[color=DarkGreen] at (2.835,1.55) {\tiny$g$};

\node[draw, color=denim, rectangle, minimum width=0.2cm, minimum height=1cm, line width=0.45pt] at (3.4,1.7) {};
\node[color=denim] at (3.4,2.35) {\tiny$q(:,a)$};


\draw[line width=0.65pt,->] (7.95,2.25) -- (6.96,2.25);
\node at (7.5,2.4) {\scriptsize Value};
\node at (7.5,2.1) {\scriptsize $v$};

\draw[line width=0.65pt,->] (5.22,2.25) -- (4.2,2.25);
\node at (4.8,2.45) {\footnotesize$\mathring{v}$};

\draw[line width=0.65pt,->] (7.95,1.25) -- (4.2,1.25);
\node at (7.5,1.4) {\scriptsize Group};
\node at (7.5,1.1) {\scriptsize $g$};

\draw[line width=0.65pt,<->] (2.7,2.25) -- (0.75,2.25);
\node at (1.75,2.4) {\scriptsize  Query Matrix};
\node at (1.75,2.1) {\scriptsize  $q$};
\draw[line width=0.65pt,->] (2.7,1.25) -- (0.75,1.25);
\node at (1.75,1.4) {\scriptsize Answer};
\node at (1.75,1.1) {\scriptsize $a$};

\node[rotate=270] at (8.2,2) {\small User};

\draw[-,line width=1.5pt] (0.75,0.75)--(0.75,3.25);
\node[rotate=90] at (0.5,2)  {\small Server};


\node[draw,rectangle,dashed, minimum width=7cm, minimum height=2.5cm,color=gray] at (4.45,2) {};

\end{tikzpicture}
}
    \caption{A block diagram representing the Q\&A scheme for a binary alphabet, $\cV=\{-1,1\}$. 
    The user is assigned a query matrix $q$.  He sends the server an answer, $a$, which is an index of a column of this matrix. His answer is based on his group, $g$, and his randomized value,  $\mathring{v}$. To randomize his value, the user applies randomized response, parameterized by $\lambda$.
    }
    \label{fig:block_QA}
\end{figure*}
    
    \item \textbf{\emph{Privacy:}} We keep a user's group private. We use local differential privacy \cite{Dwork_2006,Duchi_2013} as our measure of privacy for  a user's group. Since a user's value and group can be correlated,  it is sometimes necessary (depending on the required privacy parameter) to also hide a user's value in addition to his group. 
To that end, a user's answer to the server is the output of a randomized mechanism  $\cM : \cG\times \cV\times\cQ \rightarrow \cA$ that outputs a user's answer $a$ belonging to an alphabet $\cA$ based on his  group, value, query and local randomness.
    
    \begin{definition}
    \label{def:privacy}
    Let $\ep$ be a positive real number, and $\cM$ be a randomized mechanism.  We say $\cM$ is $\ep$-locally differentially private with respect to the group if for any $g,g'\in\cG$, 
    $q\in \cQ$, and $a\in\cA$,
    \begin{multline}
    \label{eq:privacy}
        \Pr(\cM(G,V,Q)= a|G=g,Q=q,{P}={p},\Theta=\theta)\leq \\e^\ep \Pr(\cM(G,V,Q)= a|G=g',Q=q,{P}={p},\Theta=\theta),
    \end{multline}
     where the probability is taken over the randomness of the mechanism and the random variable $V$.
    \end{definition}

    The probabilities in the local differential privacy definition are taken given the realizations of the random variables $P$ and $\Theta$. Even though the server does not necessarily know these realizations, the privacy definition above  assumes this knowledge. This is needed because,  with enough answers collected from users, the server might infer information about the distributions of $P$ and $\Theta$. 
    
   We note that if a randomized mechanism is $\ep_0$ locally differentially private, then it is also locally differentially private for all $\ep>\ep_0$ motivating the following definition. 
    \begin{definition}
    \label{def:priv_level}
    The privacy level of a scheme (randomized mechanism) is the smallest $\ep>0$ such that \eqref{eq:privacy} is satisfied. 
    \end{definition}
\end{enumerate}

\section{Main Results}
\label{sec:main_results}

\textit{\textbf{Query and Aggregate (Q\&A) Scheme:}} We propose a new scheme for PMGA, which we refer to as the  Query and Aggregate (Q\&A) scheme. Q\&A is characterized by its low communication cost per user, which is independent of the number of groups $k$ and number of users $n$. It also offers an advantageous accuracy for the high privacy regime. 
 Figure \ref{fig:block_QA} summarizes this scheme, which mainly consists of two blocks:
 \begin{enumerate}

     \item {\it  Query/Answer block:} The user is assigned a query matrix\footnote{The assigned query matrix, $q$, is independent of the user's group and value, and is known to both the server and the user.}. His answer is an index of a column of this matrix, and is determined by his value and group. This leverages the randomness in the user's value to hide his group and already provides a level of privacy over the user's group. 

     \item {\it  Randomized Response block:}
       Here, the user adds noise to his value parameterized by $\lambda\geq 0$. This block is not always necessary, except for some cases, such as when the users' groups and values are highly correlated.
       

 \end{enumerate}
 Theorem \ref{thm:QA} characterizes the communication cost, privacy, and accuracy achieved by the Q\&A scheme.

\begin{theorem}[Q\&A Scheme]
\label{thm:QA}
Given a PMGA instance with   $n$ users,  $k$ groups,  alphabet  $\cV=\{\pm 1,\dots,\pm m\}$, and  the users' value distribution $p_g(v)$ for all $g\in\cG,v\in\cV$; the Query and Aggregate  scheme (Q\&A) satisfies the following properties.

\begin{enumerate}
    \item The Q\&A scheme has a communication cost of $\log(2m)$ bits per user.
    \item The Q\&A scheme is $\ep_{\QA}$-LDP with 
    \begin{align}
    \label{eq:priv_C_general}
    e^{\ep_{\QA}}=   \max_{\mathclap{{\substack{{v,v'\in\cV},\\{g,g'\in\cG,}{g'\neq g}}}}}\left\{\frac{(2m(1-\lambda)-1)p_{g}(v)+\lambda}{(2m(1-\lambda)-1)p_{g'}(v')+\lambda}\right\} ,
    \end{align}
    where the randomization parameter $\lambda\in\left[0,\frac{2m-1}{2m}\right)$.
    \item  The estimator of the Q\&A scheme is unbiased and has  relative mean square error 
    \begin{equation}
        \label{eq:QA_error}
        \Er_\QA =\alpha n^{-1},
    \end{equation}
    where {\medmuskip=0mu \thinmuskip=0mu \thickmuskip=0mu $\alpha=\frac{2m\lambda\E[V_i^2]}{2m-2m\lambda-1}+\frac{(4m^2-1)(m+1)\left[(2m-1)(k-1)+2m\lambda\right]}{6(2m-2m\lambda-1)^2}$}. The relative mean square error is  $\cO\left(\frac{km^4}{n}\right)$. 
\end{enumerate}
\end{theorem}

\noindent
We explain the Q\&A scheme in more details in Section \ref{sec:QA_Scheme}.

\textit{\textbf{Randomized Group (RG) Scheme:}} To better understand the performance of the Q\&A scheme described in Theorem~\ref{thm:QA}, we compare it to the Randomized Group (RG) scheme which adds noise directly to the group.  With probability $\lambda_{gr}$, the user sends the server his true group ($\log(k)$ bits) and true value or a noisy version of it ($\log(2m)$ bits). Otherwise, the user lies about his group and sends a mean zero random value that is independent of his true value.

This scheme is an adaptation of the randomized response \cite{Warner_1965, Kairouz_2016} method used in the differential privacy literature. In 
Theorem \ref{thm:SchemeR} in Section \ref{sec:Scheme_R} we present the details and analysis  of the RG scheme.

\textit{\textbf{Comparison (Q\&A vs.\ RG):}} The Q\&A scheme requires $\log(2m)$ bits per user, while the RG scheme requires $\log(2km)$  bits per user. Therefore, from a user-centric perspective, the Q\&A scheme always outperforms the RG scheme in terms of communication cost. 
However, they achieve different error and privacy trade-offs. 
We also look at the communication cost from a server-centric perspective by fixing  the total  communication cost at the server, and comparing the relative error versus privacy. This allows for a different   number of users for each of the two schemes.\footnote{This is motivated by the idea that, in practice, the server might be choosing a batch of users from a larger population.}

Figure \ref{fig:comparisons}, gives an instance of this comparison for a fixed  communication cost.
The key takeaway from this comparison is that there are two regimes, 
\textit{(i) a high privacy regime} where for small values of the  privacy parameter, $\ep$,  Q\&A  outperforms  RG;
\textit{(ii) a low privacy regime} where for large enough privacy parameter, $\ep$,  RG  outperforms  Q\&A.  This is because, as $\ep$ goes to infinity, the error of the Q\&A scheme converges to a constant strictly bounded away from zero as we cannot further tune the parameters of the scheme. On the other hand, the error of the RG scheme converges to zero. 
We defer a more detailed comparison to Section \ref{sec:comparisons}.






\section{The Query and Aggregate (Q\&A) Scheme}
\label{sec:QA_Scheme}

In this section, we describe  the Q\&A scheme. We begin by an example that illustrates the key ideas of this scheme by focusing on the special case of two groups and a binary alphabet.
We then give the description of the  general (Q\&A) scheme in Section \ref{sec:gen_C}.

\subsection{1-bit Example: Two groups and a binary alphabet}
\label{sec:C_special case}

We focus on  the special case of two groups, $k=2$, and a binary alphabet, $\cV=\{-1,1\}$. In this case, the  Q\&A scheme needs only a   single bit of communication  per user.

\subsubsection*{\textbf{Scheme Description}}
\label{sec:schemeC_special_intrinsic}

 The scheme is composed of the following three steps.
\begin{enumerate}[label=\theenumi.]
    
    \item {\em Queries:} Each user $i$ responds to a random query $q_i$ which is a $2$ by $2$ matrix. More specifically, the query $q_i$ is chosen uniformly at random from the set
    $$\cQ=\left\{\left[\begin{smallmatrix} -1& +1\\ +1& -1\end{smallmatrix} \right],\left[\begin{smallmatrix} -1& +1\\ -1& +1\end{smallmatrix} \right],\left[\begin{smallmatrix} +1& -1\\ -1& +1\end{smallmatrix} \right], \left[\begin{smallmatrix} +1& -1\\ +1& -1\end{smallmatrix} \right]\right\}.$$
    The user's assigned query is independent of his group and value. Moreover, it is assumed that the server knows  the queries assigned to each user.

    \item {\em User's answer}: Each user sends the server a $1$-bit answer, $a_i$, depending on  the query he received. The user only looks at the {\em row} of the query matrix that corresponds to his group, \Ie row $1$ if he is in group $1$ and row $2$ if he is in group $2$. He answers with the index of the {\em column} that contains his value, \Ie $a_i=1$ or $a_i=2$.

    \item {\em Server's estimation:} The server receives the $1$-bit answer $a_i$ from each user $i$. He maps the $1$-bit answer into the  vector $q_i(:,a_i)$, \Ie the $a_i^\textnormal{th}$ column of the query matrix $q_i$. This is possible because he knows the user's assigned query.
    Then, the server forms the  estimates of the aggregate for each group as follows:
    \begin{equation}
    \label{eq:special_case_sum}
        \hbS_{\QA}=\sum_{i=1}^n q_i(:,a_i). 
    \end{equation}

 \end{enumerate}

For example, consider a user $i$ in group $2$ who has value $+1$. If he receives the query $q_i=\left[\begin{smallmatrix} -1& +1\\ +1& -1\end{smallmatrix}\right]$, then his answer is $a_i=1$, which the server maps into the vector  $q_i(:,a_i)=\left[\begin{smallmatrix} -1\\ +1\end{smallmatrix} \right]$. Otherwise, if the user receives the query $q_i=\left[\begin{smallmatrix} -1& +1\\ -1& +1\end{smallmatrix} \right]$, then his answer is $a_i=2$, which is mapped into $q_i(:,a_i)=\left[\begin{smallmatrix} +1\\ +1\end{smallmatrix} \right]$.

The key idea behind these queries is that they provide different, and equally likely, pairings of a value for a particular group with all possible values of the other group. For instance, if we look at the first column of the query matrices, notice that in the query $\left[\begin{smallmatrix} -1& +1\\ +1& -1\end{smallmatrix} \right]$,  the value $-1$ for group $1$ (first row) is paired with the value $+1$ of group $2$ (second row), while in query $\left[\begin{smallmatrix} -1& +1\\ -1& +1\end{smallmatrix} \right]$ it is paired with the value $-1$ of group $2$. 

Next we give a brief analysis of  this scheme, and see how it fairs on our three performance metrics: accuracy (MSE), privacy, and communication cost.

\subsubsection*{\textbf{Accuracy}}
We show that the relative mean square error goes to zero as the number of users increases, allowing the server a better estimate of the true aggregate $\bS$.

Without loss of generality, let us consider,  $\bS(1)$, the aggregate corresponding to group $1$. Then, its estimate is 
\begin{align}
\label{eq:special_case_sum_2}
     \hbS_{\QA}(1)
     &=\sum_{i\in[n]:g_i=1}  q_i(1,a_i)+\sum_{i\in[n]:g_i=2}  q_i(1,a_i)\nonumber
     \\&=\underbrace{\bS (1)}_{\substack{\text{True Aggregate}\\\text{for Group 1}}} +\underbrace{\sum_{i\in[n]:g_i=2}  q_i(1,a_i)}_\text{Noise}.
 \end{align}

Therefore, the estimate $\hbS_{\QA}(1)$ can be interpreted as   the true aggregate with an added noise term. The noise corresponds to the contribution of   the users who do not belong to group~$1$. Since the queries were assigned uniformly at random, the distribution of the answers corresponding to the noise is uniform and independent of the true aggregate $\bS(1)$. It follows from our choice of query matrices that the contribution to the estimate, of each user $i$ in group~$2$,  $ q_i(1,a_i)$, is a realization of the random variable,
\begin{align}
\label{eq:a_hat_QA_special}
     Q_i(1,A_i) = \begin{cases}
-1 & \text{with probability} \hspace{0.5em} \frac{1}{2}, \vspace{0.25em}\\ 
+1 & \text{with probability} \hspace{0.5em} \frac{1}{2}.
\end{cases}
\end{align}

The noise term can be interpreted as the position of a point on the integer number line, $\mathbb{Z}$, after $n$ steps of a simple random walk starting at zero. Alternatively, the noise is the sum of i.i.d. random variables with bounded variance that converges to a zero mean additive Gaussian noise. Either way, this implies that the expectation of the norm of the noise grows as $\mathcal{O}(\sqrt{n})$. And indicates that the relative mean square error, $\Er_\QA$, goes to zero as $\cO(n^{-1})$.

\subsubsection*{\textbf{Privacy}}

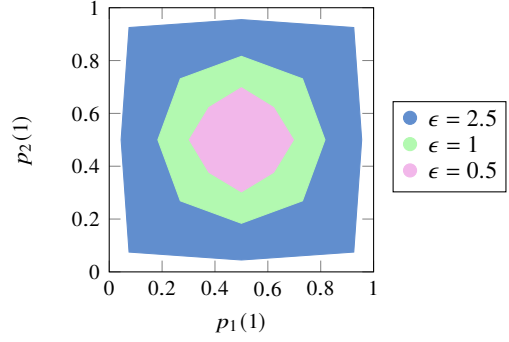
\begin{figure}
    \centering
    \definecolor{lightdenim}{RGB}{223,233,246}
\definecolor{lightgreen}{RGB}{214,232,214}
\definecolor{lavendergray}{RGB}{218,232,229}
\begin{tikzpicture}

\begin{axis}[%
width=100,
height=100,
scale only axis,
xmin=0,
xmax=1,
ymin=0,
ymax=1,
xlabel=$p_{1}(1)$,
ylabel=$p_{2}(1)$,
yticklabel style = {font=\footnotesize,xshift=0pt},
xticklabel style = {font=\footnotesize,yshift=0pt},
x label style={font=\footnotesize,yshift=3pt},
y label style={font=\footnotesize,yshift=-2pt},
]
\draw[fill= bleudefrance,color= bleudefrance] (7.5,7.5) -- (4.5,50) -- (7.5,92.5) --(50,95.5) -- (92.5,92.5) -- (95.5,50) -- (92.5,7.5) -- (50,4.5) --cycle;
\draw[fill= grannysmithapple, color= grannysmithapple] (26.9,26.9) -- (18.4,50) -- (26.9,73.1) --(50,81.6) -- (73.1,73.1) -- (81.6,50) -- (73.1,26.9) -- (50,18.4) --cycle;
\draw[fill= lavenderpurple, color= lavenderpurple] (37.7,37.7) -- (30.3,50) -- (37.7,62.2) --(50,69.7) -- (62.2,62.2) -- (69.7,50) -- (62.2,37.7) -- (50,30.3) --cycle;

\end{axis}

\node[draw,rectangle,minimum height=1.2cm,minimum width=1.5cm] at (4.525,1.665) {};
\node at (4,2.05) [circle,draw,circle, inner sep=1.75pt,fill=bleudefrance, color=bleudefrance] {};
\node at (4,1.7) [circle,draw,circle, inner sep=1.75pt,fill=grannysmithapple, color=grannysmithapple] {};
\node at (4,1.35) [circle,draw,circle, inner sep=1.75pt,fill=lavenderpurple, color=lavenderpurple] {};

\node at (4.7,2.05)  {\small $\ep=2.5$};
\node at (4.58,1.7)  {\small $\ep=1$};
\node at (4.7,1.35)  {\small $\ep=0.5$};

\end{tikzpicture}%
    \caption{Privacy for the Q\&A Scheme. The values of $p_{1}(1)$ and $p_{2}(1)$ in the shaded regions of the figures above guarantee the fixed privacy parameters $\ep=2.5$, $\ep=1$, and $\ep=0.5$, respectively. The higher the privacy requirement, \Ie smaller $\ep$, the smaller the region. We note that the indicated region is the full interior of the polygon.}
    \label{fig:privacy}
\end{figure}

We show that the Q\&A scheme is $\ep_\QA$ locally differentially private. From Definition \ref{def:privacy},
\begin{align}
\label{eq:special_privacy_analysis}
    e^{\ep_{\QA}}
    &= \max\limits_{\substack{{g,g'\in\{1,2\},}\\{{a}\in \{1,2\}, q\in\cQ}}} \frac{\Pr({A}_i={a}|G_i=g,Q_i=q)}{\Pr({A}_i={a}|G_i=g',Q_i=q)} .
\end{align}
The first thing we notice is that the ratio in \eqref{eq:special_privacy_analysis} is equal to $1$ for $g=g'$, and the maximum  is always greater than or equal to $1$ when  $g\neq g'$. Therefore, we can limit the maximization in \eqref{eq:special_privacy_analysis} to $g\neq g'$. Moreover, a user's  value ($+1$ or $-1$)  is a deterministic  function of the answer, the  query, and the  group. 
Therefore, we can simplify
 \eqref{eq:special_privacy_analysis} to
\begin{align}
    \label{eq:special_privacy_analysis_2}
    e^{\ep_{\QA}}
    &=\max\limits_{\substack{{g,g'\in\{1,2\},g\neq g'}\\{v,v'\in \{-1,1\}, q\in\cQ}}} \frac{\Pr({V}_i={v}|G_i=g,Q_i=q)}{\Pr({V}_i={v'}|G_i=g',Q_i=q)}\nonumber\\
    &=\max\limits_{\substack{{g,g'\in\{1,2\},g\neq g'}\\{v,v'\in \{-1,1\}}}} \frac{p_g(v)}{p_{g'}(v')},
\end{align}
which follows from the independence of the random variables representing the user's value, $V_i$, and his assigned query, $Q_i$, and the definition   $p_g(v)=\Pr(V_i=v|G_i=g)$.\footnote{To simplify our discussion, in the rest of this paper, we assume that the probabilities $p_g(v)$ are in $(0,1)$, for all $g\in[k]$ and $v\in\cV$. }
Thus, we obtain an expression of the privacy which only depends on the users' value distributions.

We refer to the privacy parameter $\ep_{\QA}$, described in \eqref{eq:special_privacy_analysis_2}, as the \textit{intrinsic privacy} of the scheme. Notice that it depends on the users' value distributions, $p_1(\cdot)$ and $p_2(\cdot)$; however, neither the server nor the users know these $p_1(\cdot)$ and $p_2(\cdot)$. Therefore, they cannot directly calculate the privacy parameter $\ep_{\QA}$. Nonetheless, the privacy parameter, $\ep_{\QA}$, can be bounded if the users have prior information about their value distributions. For example, suppose the users know that $p_1(1)$ and $p_2(1)$ are bounded  such that $c_{\min}\leq p_1(1),p_2(1)\leq c_{\max}$, where the constants $c_{\min},c_{\max}\in (0,1)$. In this case, we can upper bound   the intrinsic privacy level
 $
     \ep_{\QA}\leq \ln\left(\frac{c_{\max}}{c_{\min}}\right). 
 $

Next we give more insights about the relationship between the users' value distributions, $p_g(\cdot)$, and the privacy parameter.
Let us fix a privacy level $\ep$, and define the region that describes the users' value distributions, $p_1(1), p_2(1)\in (0,1)$, which guarantee that the scheme is $\ep$-LDP. 
In Figure \ref{fig:privacy}, we plot this region 
for different values of $\ep$.
 Looking at Figure \ref{fig:privacy} and \eqref{eq:special_privacy_analysis_2}, we observe the following.
\begin{itemize}
    \item The less privacy we require, \Ie the larger the privacy level $\ep$, the larger the highlighted region, \Ie more values of $p_1(1)$ and $p_2(1)$ can guarantee this level of privacy.
    \item  The closer $p_{1}(1)$ and $p_{2}(1)$ are to $0.5$, the higher the privacy guarantee.  And perfect privacy, \Ie $\ep=0$, is only guaranteed when $p_{1}(1)=p_{2}(1)=0.5$. Intuitively, this occurs because when $p_{1}(1)=p_{2}(1)=0.5$, a user's answer to the query is independent of his group.
  
\end{itemize}

A takeaway from the above observations is that  not all privacy levels can be guaranteed for fixed user value distributions $p_1(\cdot)$ and $p_2(\cdot)$.  In other words, the intrinsic privacy of the scheme may  not always  be enough. The reason is that, in its basic form, the Q\&A scheme described above, does not guarantee privacy over the user's value. Thus, when the user's value and group are sufficiently correlated, the user's value might leak more information about his group than permitted by the $\ep$-LDP requirement. In such cases, the general Q\&A scheme adds a second layer of privacy to the user's value to further hide his group. In addition, this provides flexible privacy guarantees which do not  depend only on the user's value distributions. 
This second layer of privacy is obtained by adding a randomized response block, parameterized by the probability of lying  $\lambda$, which hides a user's value (see Figure~\ref{fig:block_QA}). We give a full description of the general Q\&A scheme in Section \ref{sec:gen_C}. 


\subsubsection*{\textbf{Communication}} 
Since the user's answer is either $1$ or $2$, \Ie $a_i\in\{1,2\}$, the scheme's communication cost is one bit per user. Moreover, we show in Theorem~\ref{thm:QA} that the general scheme's communication cost is always $1$ bit per user when the alphabet, $\cV$, is binary, irrespective of the number of groups. This is the fundamental limit on the zero-error communication cost if there were no groups and no privacy requirements.

Note that the query assignment must be known to both the server and the user. This  can be accomplished without incurring  communication cost.  
For instance, it can be implemented as the output of a public hash function that takes as input the user's index $i\in[n]$, or simply considered part of the scheme agreement that does not depend on a user's group and value.

\subsection{The General Q\&A Scheme}
\label{sec:gen_C}

\begin{figure}[t]
    \centering
    \begin{tikzpicture}
\definecolor{fireenginered}{rgb}{0.81, 0.09, 0.13}
\definecolor{denim}{rgb}{0.08, 0.38, 0.74}
\definecolor{lightdenim}{RGB}{223,233,246}
\definecolor{lightgreen}{RGB}{214,232,214}
\definecolor{lavendergray}{RGB}{218,232,229}

\node[fill, dashed, color=lightgreen, rectangle, minimum width=3.25cm, minimum height=0.75cm, line width=1pt] at (3,1.25) {};
\node[fill, dashed, color=lightdenim, rectangle, minimum width=0.75cm, minimum height=3.25cm, line width=1pt] at (3.5,2) {};

\node[color=black] at (1,2) {$q_i=$};
\draw[-,color=black, line width=0.5pt] (1.75,0.5) -- (1.5,0.5) -- (1.5,3.5) -- (1.75,3.5);
\draw[-,color=black, line width=0.5pt] (4.25,0.5) -- (4.5,0.5) -- (4.5,3.5) -- (4.25,3.5);


\node[fill, color=lavendergray, rectangle, minimum width=0.75cm, minimum height=0.75cm, line width=1pt] at (3.5,1.25) {};
\node[color=black] at (3.5,1.25) {$\mathring{v}_i$};

\node[draw, dashed, color=DarkGreen, rectangle, minimum width=3.25cm, minimum height=0.75cm, line width=1pt] at (3,1.25) {};
\node[color=DarkGreen] at (4.9,1.25) {$g_i$};

\node[draw, dashed, color=denim, rectangle, minimum width=0.75cm, minimum height=3.25cm, line width=1pt] at (3.5,2) {};
\node[color=denim] at (3.5,3.9) {$q_i(:,a_i)$};

\end{tikzpicture}
    \caption{User $i$ sends the answer $a_i$ based on  his assigned query matrix $q_i$, his group $g_i$, and his randomized value $\mathring{v}_i$. His answer is the index of the column that contains his randomized value. The server maps this answer to the $a_i^{th}$ column of $q_i$.}
    \label{fig:query_matrix}
\end{figure}
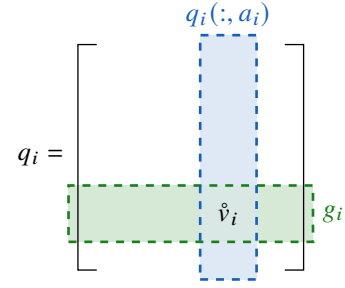

In this section, we describe the general Q\&A scheme, for any number of groups $k\geq 2$, and alphabet parameter $m\in\mathbb{N}$.
This scheme is obtained by generalizing the query matrices of the previous example and is presented in 
Figure \ref{fig:query_matrix}. 
The  Q\&A scheme includes an additional randomized response block for improved privacy as described in Figure~\ref{fig:block_QA}. 

\begin{enumerate}[wide, labelindent=0pt]
\item \textbf{Queries:} 
Each user is assigned a random query matrix of dimension   $k\times 2m$  and elements in $\cV=\{\pm 1,\dots,\pm m\}$. The query matrices assigned to each user are chosen independently and uniformly at random from the set $\cQ$ defined as
\begin{equation}
\label{eq:def_cQ}
    \cQ:=\left\{q\in\cV^{k\times 2m}\big| 
q(g,:)\in \textbf{Sym}(\cV) \textnormal{ for all } g\in [k]\right\},
\end{equation}
where $q(g,:)=(q(g,1),\dots,q(g,k))$ and $\textbf{Sym}(\cV)$ is the set of all row vectors which are an ordered permutation of the finite set $\cV$.\footnote{An ordered permutation of a set $\cV$ is a vector where each element is a distinct element of $\cV$, \Eg $\textbf{Sym}(\{\pm 1,\pm 2\})$ has $4!=24$ elements including $(-2,-1,1,2)$ and $(1,2,-2,-1)$. }
Each row of a matrix $q\in\cQ$ is a permutation of all the possible $2m$ values. Notice that the values cannot be repeated within a  row but rows can be repeated.  We denote by $q_i$ the query assigned to user $i$.

We assume that the server also knows the query assigned to the user. As previously mentioned, since the query does not depend on the user's group or value, it can be assigned offline as part of the scheme agreement, or  implemented as the output of a public hash function.



\item \textbf{User's Answer:} Given his assigned query, user $i$ hides his value using the randomized response block parameterized by $\lambda$, as in Figure \ref{fig:block_QA}.  That is, given his true value $v_i$, the user first chooses a randomized value $\mathring{v}_i$ according to the distribution
\begin{equation}
\label{eq:dist_V_ring}
    \Pr(\mathring{V}_i=\mathring{v}_i|V_i=v_i)=\begin{cases}
    1-\lambda &\text{ for } \mathring{v}_i=v_i\\ \frac{\lambda}{2m-1} &\text{ for } \mathring{v}_i\in\cV-\{v_i\},
    \end{cases}
\end{equation}
where $\lambda\in[0,1-\frac{1}{2m})$.
When $\lambda=0$, \Ie no privacy over the user's value, then $V_i=\mathring{V}_i$.

Then, user $i$ looks at the $g_i^{\textnormal{th}}$ row ($g_i$ is the user's group) of the query matrix $q_i$, and sends to the server the answer $a_i$, which is  the index of the column that has his randomized value $\mathring{v}_i$.
More precisely, $a_i$ is such that $q_i(g_i,a_i)=\mathring{v}_i$ as explained in Figure~\ref{fig:query_matrix}.

\item \textbf{Server's Estimation:} Upon receiving user $i$'s answer, the server maps it into the $a_i^{\textnormal{th}}$ column of query $q_i$, \Ie
$$q_i(:,a_i)=(q_i(1,a_i),q_i(2,a_i),\dots , q_i(k,a_i))^\top.$$ 
The server sums the mapped answers from all the users, and multiplies by an unbiasing term (see Appendix \ref{app:proof_thm_1} for more details), to find the estimate of the true aggregate, $\bS$, \Ie  \begin{equation}
\label{eq:server_estimation_QA}
    \hbS_{\QA}=\frac{2m-1}{2m-2m\lambda-1}\sum_{i\in[n]}q_i(:,a_i).
\end{equation} 

\end{enumerate}

Below we give examples of possible queries and answers.

\begin{example}
Consider the setting where there are $k=3$ groups and the alphabet of values is $\cV=\{\pm1,\pm 2\}$. Let $\lambda=0$, \Ie $V_i=\mathring{V}_i$. Suppose that user $1$ has value $v_1=-1$ and belongs to group $g_1=2$. For instance, if he is assigned the query 
$$
    q_1=\left[\begin{smallmatrix}
    -2&-1&+1&+2\\
    -2&+1&-1&+2\\
     +2&-1&-2&+1
    \end{smallmatrix}
    \right],
$$
then his answer is $a_1=3$, because his value, $v_1=-1$, is the third element of the second row (corresponding to his group $g_1=2$) of $q_1$. Upon receiving this answer, the server decodes it into the third column of $q_1$, \Ie  $q_1(:,a_1)=(+1,-1,-2)^\top$. 

If the user is assigned the query
$$
    q_1=\left[\begin{smallmatrix}
    -2&-1&+1&+2\\
    +1&-2&+2&-1\\
     +1&-2&+2&-1
    \end{smallmatrix}
    \right],
$$
his answer will be $a_1=4$, which the server decodes into $q_1(:,a_1)=(+2,-1,-1)^\top$. In both cases $q_1(g_1,a_1)=v_1$.
\end{example}

We note a few characteristics of this design of queries and answers. Since every row of any query matrix $q\in\cQ$ contains all possible values, the user's value is always  one of the elements of the row vector corresponding to his group. Moreover, from the server's perspective, looking at the mapped answer $q_i(:,a_i)$, \Ie a column vector of the user's assigned query $q_i$, the user's value (or randomized value) is in row $g_i$ of $q_i$. As for all the other elements of the vector, they are uniformly distributed over $\cV$. This follows from the design of the query alphabet $\cQ$ and mirrors \eqref{eq:a_hat_QA_special} from the previous section. It is also the key idea for the accuracy proof of Theorem \ref{thm:QA}.



An interesting property of the Q\&A scheme is that, depending on the required privacy, one can choose $\lambda=0$, \Ie no randomized response block in Figure \ref{fig:block_QA}. The Q\&A scheme still guarantees local differential privacy with
\begin{equation}
\label{eq:intrinsic_privacy_QA}
    \ep_{\QA}= \ln \left( \max_{\substack{{v,v'\in\cV}\\{g,g'\in\cG,g'\neq g}}}\left\{\frac{p_{g}(v)}{p_{g'}(v')}\right\} \right),
\end{equation}
which follows from \eqref{eq:priv_C_general}. 
As in the previous section, we refer to  this as the  \textit{intrinsic privacy} of the scheme,  which corresponds to the special case of $\lambda=0$. 
If the intrinsic privacy is not enough because of a high correlation between the user's group and value, external noise can be added to the values through the randomized response block with $\lambda$ chosen appropriately depending on the required privacy $\ep$.


\begin{remark}[The choice of $\lambda$]
\label{rem:choice_of_rho}
Given a required privacy parameter $\ep$, the parameter $\lambda$ that can guarantee this given $\ep$ is determined using \eqref{eq:priv_C_general}. However, this requires the knowledge of the value distributions, $p_g(.)$ for all $g\in\cG$.
Nevertheless, one can still use \eqref{eq:priv_C_general} to find a bound on  $\lambda$  that is independent of the users' value distributions as follows,
\begin{equation}
\label{eq:QA_rho}
    \lambda\geq\frac{2m-1}{2m+e^\ep -1}.
\end{equation}
This bound can be tightened if some side information is known about the users' value distributions. For instance, suppose that $c_{\min}<p_g(v)< c_{\max }$ for all $g\in\cG$ and $v\in\cV$, for some  constants $c_{\max}, c_{\min} \in [0,1]$, $c_{\max}> c_{\min}$. In this case, the following tighter bound can be shown 
\begin{equation}
\label{eq:QA_rho_2}
    \lambda\geq\frac{(2m-1)c_{\max}-c_{\min}e^\ep}{2m(c_{\max}-c_{\min}e^\ep)+e^\ep -1}.
\end{equation}
 Evidently, smaller values of $\lambda$ are better for accuracy because the mean square error is increasing in $\lambda$.
 
\end{remark}

\begin{remark}[Error Calculation]
Computing the mean square error relies on generalizing the  approach in the example in Section \ref{sec:C_special case}. We have two types of errors in the estimate of the aggregate per group, \Ie  $\hbS_{\QA}(g)$. The first  is the error introduced by the users that are not in  group $g$. This can be approximated by a  zero mean additive noise as shown in  \eqref{eq:special_case_sum_2}. The second  is the error introduced by the randomized response block acting on  the users' value. This error biases the sum $\sum q_i(:,a_i)$. Therefore, to unbias the estimator we multiply by  $({2m-1})/({2m-2m\lambda-1})$ as seen in \eqref{eq:server_estimation_QA}.
The details of the error calculation can be found in the proof  in  Appendix \ref{app:proof_thm_1}.
\end{remark}

Theorem \ref{thm:QA}, first stated in Section \ref{sec:main_results}, provides the performance of the Q\&A scheme with respect to communication cost, privacy, and accuracy. For its proof see Appendix \ref{app:proof_thm_1}.


\section{The Randomized Group (RG) Scheme}
\label{sec:Scheme_R}
To better gauge the performance of the Q\&A scheme we compare it to the Randomized Group (RG) scheme which directly hides a user's group by adding noise to it through a randomized response step. 
In RG, each user $i$ sends the server an answer $a_i=(\mathring{g}_i,\mathring{v}_i)$ of his privatized group and value.
That is, $\mathring{g}_i$ is chosen randomly according to the distribution  
\begin{equation}
\label{eq:RG_def_G_round}
    \Pr(\mathring{G}_i=\mathring{g}_i|G_i=g_i)=\begin{cases}
    \scalebox{0.8}{$1-\lambda_{gr}$}& \textnormal{for } \mathring{g}_i=g_i\\
    \frac{\lambda_{gr}}{k-1}& \textnormal{for }\mathring{g}_i\in[k]-\{g_i\},
    \end{cases}
\end{equation}
where $g_i$ is user $i$'s group and the parameter $\lambda_{gr}\in(0,1)$. 
As for the value $\mathring{v}_i$, there are two cases: 
\begin{enumerate}[leftmargin=0.45cm]
    \item  $\mathring{g}_i\neq g_i$: In this case, the user chooses $\mathring{v}_i$, uniformly at random, \Ie
    \begin{equation}
    \label{eq:RG_def_V_round_neq}
        \Pr(\mathring{V}_i=\mathring{v}_i|\mathring{g}_i\neq g_i)=\frac{1}{2m}
    \end{equation}
        for all $\mathring{v}_i\in\cV$.
    This choice ensures  that when   users lie about their groups,  the aggregate of their contribution has a zero mean.
    \item  $\mathring{g}_i=g_i$: In this case, the user lies about his true value with probability $\lambda_{vl}\in\left[0,1-\frac{1}{2m}\right)$. That is, he randomly chooses a value, $\mathring{v}_i$, according to the distribution 
    \begin{align}
    \label{eq:RG_def_V_round}
        \Pr(\mathring{V}_i=\mathring{v}_i|V_i=v_i,&\mathring{g}_i=g_i)\hspace{-2pt}=\hspace{-2pt}\begin{cases}
        \scalebox{0.8}{$1-\lambda_{vl}$} &  \hspace{-4pt} \textnormal{for } \mathring{v}_i=v_i,\\
        \frac{\lambda_{vl}}{2m-1} & \hspace{-4pt} \textnormal{for } \mathring{v}_i\in\cV-\{v_i\}.
        \end{cases}
    \end{align}
\end{enumerate}  

The server aggregates the received answers and re-scales the aggregate to unbias the estimator, such that, for all $g\in [k]$ the estimate of the true aggregate of group $g$, $\bS(g)$, is
\begin{equation*}
    \hbS_\RG(g):=\frac{2m-1}{(1-\lambda_{gr})(2m(1-\lambda_{vl})-1)}\sum_{i:a_i(1)= g} a_i(2).
\end{equation*}
Note that there are no queries assigned to users in this scheme.

Theorem \ref{thm:SchemeR} characterizes the scheme's performance with respect to communication cost, privacy, and accuracy.

\begin{theorem}
\label{thm:SchemeR}
Given a PMGA instance with   $n$ users,  $k$ groups,  alphabet  $\cV=\{\pm 1,\dots,\pm m\}$, and  the users' value distribution $p_g(v)$ for all $g\in\cG,v\in\cV$;
the Randomized Group scheme (RG) is parameterized by the randomization parameters $\lambda_{gr}\in(0,1)$, and $\lambda_{vl}\in\left[0,1-\frac{1}{2m}\right)$ and satisfies the following properties.

\begin{enumerate}
    
    \item The RG scheme has a communication cost of $\log(2km)$ bits per user.
    \item The RG scheme is $\ep_\RG$-LDP with 
    \begin{align}
    \label{eq:priv_RG}
    e^{\ep_\RG}= \max\bigg\{\beta_1(p_{\max}\beta_2+\lambda_{vl}),\frac{1}{\beta_1(p_{\min}\beta_2+\lambda_{vl})}\bigg\},
    \end{align}
    where $p_{\max}=\max\limits_{g\in\cG,\in v\in\cV}p_{g}(v)$,  $p_{\min}=\min\limits_{g\in\cG,\in v\in\cV}p_{g}(v)$, $\beta_1=\frac{2m(k-1)(1-\lambda_{gr})}{(2m-1)\lambda_{gr}}$, and $\beta_2=(2m(1-\lambda_{vl})-1)$.
    \item 
    The estimator of the RG scheme is unbiased and has relative mean square error 
    \begin{equation}
    \label{eq:RG_error}
        \Er_\RG= \beta_3 n^{-1},
    \end{equation}
    where $\beta_3=\E[V_1^2]\left(\frac{2m-1}{(1-\lambda_{gr})(2m(1-\lambda_{vl})-1))}-1\right)+\frac{(4m^2-1)(m+1)(2m\lambda_{vl}(1-\lambda_{gr})+\lambda_{gr}(2m-1))}{6(1-\lambda_{gr})^2(2m(1-\lambda_{vl})-1)^2}$. The relative mean square error is  $\cO\left(\frac{m^4k^2}{ne^\ep}\right)$. 
\end{enumerate}

\end{theorem}

\begin{proof}
See Appendix \ref{app:proof_thm_RG}.
\end{proof}

{

The following corollary characterizes the relationship between the randomization parameters, $\lambda_{gr}$, $\lambda_{vl}$, and the privacy parameter, $\ep_\RG$.




\begin{corollary}
\label{cor:minimum_RG}
Let $\ep>0$ be the required privacy, 
then, the parameters $\lambda_{gr}^*$ and $\lambda_{vl}^*$  that guarantee the required privacy and minimize the relative error of the RG scheme are given below.
\begin{itemize}
    \item If $e^{2\ep}<\frac{p_{\max}}{p_{\min}}$ and $p_{\max}\neq \frac{1}{2m}$, then $$\textstyle \lambda_{vl}^*=\frac{(2m-1)(p_{\max}-e^{2\ep}p_{\min})}{(2p_{\max}-1)+(1-2p_{\min})e^{2\ep}}$$ and
    $$\textstyle \lambda_{gr}^*=\frac{2m(k-1)(p_{\max}-p_{\min})e^\ep}{2m(k-1)(p_{\max}-p_{\min})e^\ep+(1-2p_{\min})e^{2\ep}+2p_{\max}-1}.$$
    \item If $e^{2\ep}\geq\frac{p_{\max}}{p_{\min}}$ or $p_{\max}= \frac{1}{2m}$, then $\lambda_{vl}^*=0$ and $$\lambda_{gr}^*=\frac{2m(k-1)p_{\max}}{2m(k-1)p_{\max}+e^\ep}.$$ 
\end{itemize}
\end{corollary}
\begin{proof}
See Appendix \ref{app:comp_A_C}.
\end{proof}

The above Corollary describes the choice of parameters $\lambda_{vl}$ and $\lambda_{gr}$ that minimize the error for a given required privacy $\ep>0$. It also shows that for a high privacy requirement, \Ie privacy parameter $\ep<\frac{1}{2}\ln\left(\frac{p_{\max}}{p_{\min}}\right)$, the parameter $\lambda_{vl}$ cannot be zero.
 Intuitively, since the user's value and group are correlated,  applying a privacy preserving mechanism only over the group is not always enough. This is similar to what we had with the Q\&A scheme. The parameter $\lambda_{vl}$
characterizes the second layer of privacy that hides the user's value.

\section{Comparison of the RG and Q\&A Schemes}
\label{sec:comparisons}

The  Q\&A scheme has a  communication cost  of $\log(2m)$ bits per user, \Ie it does not depend on the number of groups $k$.
However, the communication cost of the RG scheme is $\log(2km)$ bits per user. Thus, the Q\&A scheme outperforms the RG scheme in terms of communication cost per user.


To compare the two schemes on all fronts, 
  we fix the total communication cost,  \Ie the number of bits communicated by all the users to the server, and compare the privacy vs. accuracy trade-offs. 
  We choose the parameter $\lambda$ of the Q\&A scheme that guarantees the required (given) privacy parameter $\ep$ (see Remark \ref{rem:choice_of_rho}) and minimizes the error. We find this parameter $\lambda$ by solving the following optimization problem numerically
\begin{equation*}
\begin{aligned}
& \underset{\lambda}{\text{minimize}}
& & \Er_\QA\\
& \text{subject to}
& & e^\ep\hspace{-1.5pt}\leq\hspace{-1.5pt} \max_{\substack{{v,v'\in\cV},\\{g,g'\in\cG,g'\neq g}}}\left\{\frac{(2m(1-\lambda)-1)p_{g}(v)+\lambda}{(2m(1-\lambda)-1)p_{g'}(v')+\lambda}\right\} ,\\
&&&0\leq\lambda<1-\frac{1}{2m},
\end{aligned}
\end{equation*}
where $\Er_\QA$ is from \eqref{eq:QA_error}.
Similarly, for the RG scheme, we choose the parameters, $\lambda_{gr}$ and $\lambda_{vl}$, as in Corollary \ref{cor:minimum_RG}.
  
Figure \ref{fig:comparisons}  illustrates this comparison for a fixed total communication cost.
Typically for a high enough privacy constraint, Q\&A outperforms RG, while for a low enough privacy constraint, RG outperforms Q\&A. Thus, we have two privacy regimes, a high privacy regime where  Q\&A  should be used, and a low privacy regime where  RG  should be used. 
These observations are made rigorous in Theorem \ref{thm:QA_vs_RG} below. 

 We begin by expressing the relative mean square error as a function of the privacy parameter $\ep$ for a fixed total communication cost of $b\geq 2$ bits. Since the Q\&A scheme's communication cost per user is $\log(2m)$ bits, its number of users is given by $n_\QA:= b/\log(2m)$. Analogously, the number of users for the RG scheme is given by $n_\RG:=b/\log(2km)$. 
 \footnote{We assume that the parameter $b$ is chosen such that $n_\QA,n_\RG\in \mathbb{N}$.}  
 Therefore, we normalize each scheme's mean square error by its respective number of users (squared), and we express 
\begin{equation}
\label{eq:QA_ER_ep_B}
    \Er_\QA(\ep,b)=\MSE(\hbS_\QA)n_\QA^{-2}=\MSE(\hbS_\QA)\left(\frac{\log(2m)}{b}\right)^{2},
\end{equation} 
and
\begin{equation}
\label{eq:RG_ER_ep_B}
    \Er_\RG(\ep,b)=\MSE(\hbS_\RG)n_\RG^{-2}=\MSE(\hbS_\RG)\left(\frac{\log(2km)}{b}\right)^{2}.
\end{equation}
With this notation we present the following theorem.

\begin{theorem}
\label{thm:QA_vs_RG}
Let $\cV=\{\pm 1,\dots,\pm m\}$ be the alphabet of values, and $\cG=[k]$ be the set of possible groups, and fix the total communication cost $b\in \{x\in\mathbb{N}| x/\log(2m), x/\log(2km)\in \mathbb{N}\}$. Unless $k=2$, $m=1$ and $p_1(v)= p_2(v')\neq 0.5$ for all $v,v'\in \{-1,1\}$; then, there exists, 
\begin{enumerate}[label={(\roman*)}]
    \item an $\ep_h>0$, such that for all $\ep<\ep_h$, the relative error $\Er_\QA(\ep,b)<\Er_\RG(\ep,b)$, and 
    \item an $\ep_\ell>0$, such that for all $\ep>\ep_\ell$, the relative error $\Er_\QA(\ep,b)>\Er_\RG(\ep,b)$.
\end{enumerate}
\end{theorem}
\begin{proof}
See Appendix \ref{app:comp_A_C}.
\end{proof}

For the special case of $k=2$ groups, and a binary alphabet of values, \Ie $\cV=\{-1,1\}$, and  $p_1(v)= p_2(v')\neq 0.5 $  for some $v,v'\in\cV$, there exists an $\ep_0>0$ such that for all $\ep<\ep_0$, the difference in relative errors $\Er_\QA(\ep_0,n_\QA)-\Er_\RG(\ep_0,n_\RG)=\frac{1}{b}$, where $b$ is the total communication cost. 



}


\section{Conclusion}
\label{sec:conclusion}
In this paper, we formulated the problem of private multi-group  aggregation where the goal was to privately aggregate the users' values per group. Moreover, we used local differential privacy as our measure of privacy for a user's group. We characterized two schemes: Q\&A and RG. The Q\&A scheme generally outperformed the RG scheme, in terms of privacy vs. accuracy, in the high privacy regime.

Future work for this problem includes finding theoretic bounds characterizing the best performance achievable for a given privacy and total communication cost. 
Another direction would involve mapping a larger alphabet of values $\cV$ to a smaller alphabet $\cV'$ to reduce communication costs.

\appendices
\section{The Q\&A Scheme: Proof of Theorem \ref{thm:QA}}

\label{app:proof_thm_1}
We separate the proof into three parts starting with communication, then privacy, and finally with the accuracy.
\begin{enumerate}[wide, labelindent=0pt]
\item \textbf{Communication:} The user sends the server the index of a column of the query matrix. Since the query matrix has dimension $k\times 2m$, the user sends $\log(2m)$ bits to the server. 

\item \textbf{Privacy:} From Definition \ref{eq:privacy}, the privacy of user $i$ is
\begin{align}
\label{eq:schemeC_general_privacy_proof}
    e^{\ep_{\QA}}&  =\max\left\{\max\limits_{\substack{{g,g'\in\cG,g\neq g',}\\ {{a}\in[2m],q\in\cQ}}} \frac{\Pr\left({A}_i={a}\big|G_i=g,Q_i=q\right)}{\Pr\left({A}_i={a}\big|G_i=g',Q_i=q\right)},1\right\}\nonumber\\
    & =\max\limits_{\substack{{g,g'\in\cG,g\neq g',}\\ {{a}\in[2m],q\in\cQ}}} \frac{\Pr\left({A}_i={a}\big|G_i=g,Q_i=q\right)}{\Pr\left({A}_i={a}\big|G_i=g',Q_i=q\right)},
\end{align}
where $\cQ$ is as defined in \eqref{eq:def_cQ}. Notice that if $g=g'$, the ratio of probabilities is equal to $1$, and if $g\neq g'$, the maximum of the ratio of probabilities is greater than or equal to $1$. Consider 
\begin{align}
    &\Pr\big({A}_i={a}\big|G_i=g,Q_i=q\big)\nonumber\\
    &\utag{a}{=}\sum_{v\in\cV}\Pr(V_i=v\big|G_i=g)\sum_{\mathring{v}\in\cV}\Pr\left(\mathring{V}_i=\mathring{v}\big|V_i=v\right)\nonumber\\
    &\hspace{90pt}\Pr\left({A}_i={a}\big|G_i=g,Q_i=q,\mathring{V}_i=\mathring{v}\right)\nonumber\\
    &\utag{b}{=}\sum_{v\in\cV}p_g(v)\Pr\left(\mathring{V}_i={v}^*\big|V_i=v\right)\nonumber\\
    &\textstyle\utag{c}{=}(1-\lambda)p_g(v^*)+\frac{\lambda}{2m-1}(1-p_g(v^*)),\label{eq:proof_prob_A}
\end{align}
where \uref{a} follows from the law of total probability and the random variable relationships.
As for \uref{b}, it follows from definition $p_g(v):=\Pr({V}_i={v}|G_i=g)$, and noticing that given a user's randomized value $\mathring{V}_i$, his group $G_i$, and assigned query $Q_i$, the user's answer $A_i$ is deterministic. So, the probability $\Pr\left({A}_i={a}\big|G_i=g,Q_i=q,\mathring{V}_i=\mathring{v}\right)=1$ only for one realization of  $\mathring{V}_i$ which we denote by ${v}^*=q(g,a)$, otherwise $\Pr\left({A}_i={a}\big|G_i=g,Q_i=q,\mathring{V}_i=\mathring{v}\right)=0$. Finally, \uref{c} follows from \eqref{eq:dist_V_ring}.
Substituting \eqref{eq:proof_prob_A} in \eqref{eq:schemeC_general_privacy_proof}, we obtain
\begin{align*}
    e^{\ep_{\QA}}= \max_{\substack{{g,g'\in\cG,g'\neq g,}\\ {v,v'\in\cV}}}\left\{\frac{(2m(1-\lambda)-1)p_{g}(v)+\lambda}{(2m(1-\lambda)-1)p_{g'}(v')+\lambda}\right\},
\end{align*}
where we replaced $v^*,{v'}^*\in\cV$ by $v,v'\in\cV$.

\item \textbf{Accuracy:} We start by finding probabilities relating to the user's assigned queries. User $i$ is assigned query $Q_i=q$, which is chosen uniformly at random from the set $\cQ$ defined in \eqref{eq:def_cQ}. Therefore, for fixed row $j$ and column $a$, the probability $\Pr(Q_i(j,a)=v)=\frac{1}{2m}$ for all $v\in\cV$. Note that if user $i$'s answer is $A_i$, and given his assigned query, the server maps the user's answer into the vector $Q_i(:,A_i)$. 
Given user $i$'s group $G_i=g$ and group $V_i=v$, we find the distribution of $Q_i(j,A_i)$ for all $j\in[k]$. 

That is, for all $j\neq g$, $j,g\in\cG$, and $v,v'\in\cV$, we have
\begin{align*}
    &\Pr\left(Q_i(j,A_i)=v'\big|G_i=g,V_i=v\right)=\frac{1}{2m},
\end{align*}
Otherwise, for all $j=g$, $j,g\in\cG$, and $v,v'\in\cV$,
\begin{multline}
\label{eq:q_equal}
    \Pr\left(Q_i(g,A_i)=v'\big|G_i=g,V_i=v\right)\\=\begin{cases}
    1-\lambda &\textnormal{for } v'=v\\
    \frac{\lambda}{2m-1} &\textnormal{for } v'\in\cV-\{v\}.
    \end{cases}
\end{multline}

For all $i\in[n]$, we introduce the auxiliary random variables $X_i$ and $Y_i$ for ease of notation.
For all $i\in[n]$, user $i$'s group $G_i$ and his value $V_i$ are random variables as described in Section \ref{sec:problem_formulation}. We  define an auxiliary random variable $X_i$ that functions as an indicator for both the user's group and value. More precisely,  $X_i$ is a random $k$ dimensional vector (where $k$ is the number of groups), such that $X_i(j)=0$ if $j\neq G_i$ and $X_i(j)=V_i$ if $j=G_i$.

Then one readily obtains, for all $j\in[k]$,
\begin{equation}
\label{eq:def_X}
    \Pr(X_i(j)=v)=\begin{cases}
    \theta\ell p_j(v)& \text{ for } v\in\cV,\\
    1-\theta_j& \text{ for } v=0,
    \end{cases}
\end{equation}
and, 
\begin{equation}
\label{eq:exp_sum_X_i}
    \E\left[X_i(j)\right]=\sum_{v\in\cV} vp_j(v)\theta_j=\E[V_1|G_1=j]\theta_j.
\end{equation}
Since the $X_i$'s are i.i.d. for all $i\in[n]$, we have 
\begin{multline}
\label{eq:sum_X_i}
    \E\left[\left(\sum_{i\in[n]}X_i(j)\right)^2\right]\\=n\E[V_1^2|G_1=j]\theta_j+(n^2-n)\E[V_1|G_1=j]^2\theta_j^2.
\end{multline}

For every user $i\in[n]$, we define an auxiliary random variable $Y_i=Q_i(:,A_i)$, which is a  $k$ dimensional random vector. Given user $i$'s group $G_i=g_i$ and his value $V_i=v_i$, the ${g_i}^\textnormal{th}$ coordinate of the vector $Y_i$ contains user $i$'s randomized value. All the other coordinates of the vector $Y_i$ are randomly chosen from the alphabet $\cV$. More precisely,
\begin{multline}
\label{eq:dist_Y_i}
    \Pr(Y_i(j)=v|G_i=g_i,V_i=v_i)\\=\begin{cases}
    1-\lambda& \text{ for } v=v_i \text{ and } j=g_i,\\
    \frac{\lambda}{2m-1}& \text{ for } v\in\cV-\{v_i\} \text{ and } j=g_i,\\
    \frac{1}{2m}& \text{ for } v\in\cV \text{ and } j\neq g_i,\\
    \end{cases}
\end{multline}
where $\lambda\in\left[0,1-\frac{1}{2m}\right)$. Then, following from \eqref{eq:dist_Y_i}, we obtain,
$
    \Pr(Y_i(j)=v)=\theta_j\left((1-\lambda)  p_j(v)+\frac{(1-p_j(v))\lambda}{2m-1}\right)+
    \frac{(1-\theta_j)}{2m}.
$
Then, 
\begin{align}
     \E\left[Y_i(j)\right] =\frac{2m-2m\lambda-1}{2m-1}\E[V_1|G_1=j]\theta_j, \label{eq:lem_proof_E_1}
\end{align}
Since $Y_1,\dots,Y_n$ are i.i.d., we have 
\begin{multline}
\label{eq:proof_exp_Y_i_squared}
    \E\left[\left(\sum_{i\in[n]}Y_i(j)\right)^2\right]= \frac{n(1-\theta_j)}{2m}+\frac{n\lambda\theta_j}{2m-1}\sum_{v\in\cV}v^2
    \\+(n^2-n)\left(\frac{2m-2m\lambda-1}{2m-1}\theta_j\E[V_1|G_1=j]\right)^2
    \\+n\E[V_1^2|G_1=j]\theta_j\left(\frac{2m-2m\lambda-1}{2m-1}\right).
\end{multline}
Note that $\sum_{v\in\cV}v^2=\frac{1}{3}m(m+1)(2m+1)$.

One readily obtains 
$
    \E\left[\hbS_\QA-\bS\right]=0
$
by substituting \eqref{eq:exp_sum_X_i} and \eqref{eq:lem_proof_E_1}, and observing that $X_1,\dots,X_n$ are i.i.d. and $Y_1,\dots,Y_n$ are i.i.d.. Then the estimator $\hbS_\QA$ is unbiased.

Next we calculate the expectation $\E\left[\sum_{i=1}^nY_i(j)\sum_{l=1}^nX_l(j)\right]$ which will be helpful later in the proof. Notice that given user $i$'s group, $G_i=g_i$, and value, $V_i=v_i$, the product $X_i(g_i)Y_i(g_i)$ is equal to $v_i^2$ with probability $\lambda$, and equal to $v_iv$ with probability $\frac{1-\lambda}{2m-1}$ for all $v\in\cV-\{v_i\}$. And since $X_i(j)=0$ for all $j\neq g_i$, then $X_i(j)Y_i(j)=0$ for all $j\neq g_i$. Following these observations,
\begin{align}
    \E[X_i(j)Y_i(j)]=\frac{2m-2m\lambda-1}{2m-1}\theta_j \E[V_i^2|G_i=g],\label{eq:proof_X_i_times_Y_i}
\end{align}
which follows from $\sum\limits_{v'\in\cV-\{v\}}v'=-v$. Then,
\begin{align}
    & \E\left[\sum\limits_{i=1}^nY_i(j)\sum\limits_{l=1}^nX_l(j)\right]\nonumber\\
    & \utag{a}{=}\sum\limits_{i,l\in[n],i\neq l}\E\left[Y_i(j)\right]\E\left[X_l(j)\right]+\sum\limits_{i=1}^n \E\left[Y_i(j)X_i(j)\right] \nonumber\\
    &\utag{b}{=}\frac{2m-2m\lambda-1}{2m-1}\Big[(n^2-n)\theta_j^2\E[V_1|G_1=j]^2\nonumber
    \\&\hspace{120pt}+n\theta_j\E[V_1^2|G_1=j]\Big].\label{eq:proof_X_i_times_Y_i_2}
\end{align}
We have that \uref{a} follows from that fact that if $i\neq l$, then $Y_i(j)$ is independent of $X_l(j)$. And \uref{b} follows from substituting \eqref{eq:exp_sum_X_i}, \eqref{eq:lem_proof_E_1}, and \eqref{eq:proof_X_i_times_Y_i}. Then,
\begin{multline}
    \Er_{\QA}=\frac{2m\lambda}{n(2m-2m\lambda-1)}\E[V_1^2]\\
   +\frac{(4m^2-1)(m+1)}{6n(2m-2m\lambda-1)}\left(\frac{(2m-1)k}{2m-2m\lambda-1}-1\right),\label{eq:C_MSE_genera_proof}
\end{multline}
follows from substituting \eqref{eq:sum_X_i}, \eqref{eq:proof_exp_Y_i_squared}, and \eqref{eq:proof_X_i_times_Y_i_2} in \eqref{eq:error_def}.

Now we have the exact expression of $\Er_{\QA}$ in terms of $\lambda$. To characterize the accuracy vs. privacy trade-off, one might be interested in the expression of $\Er_{\QA}$ in terms of the privacy parameter $\ep$.
We can get a loose upper bound on $\Er_{\QA}$, as a function of $\ep$, by substituting $\lambda=(2m-1)/(2m-1+e^\ep)$, from Remark \ref{rem:choice_of_rho}, in \eqref{eq:C_MSE_genera_proof}, and we get $\Er_{\QA}=\cO\left(\frac{km^4}{n}\right)$.

\end{enumerate}


\section{The Randomized Group (RG) Scheme}
In Section \ref{app:proof_thm_RG} of this appendix we prove Theorem \ref{thm:SchemeR}, and in Section \ref{app:proof_co_minimum_RG} we prove Corollary \ref{cor:minimum_RG}.

\subsection{Proof of Theorem \ref{thm:SchemeR}}
\label{app:proof_thm_RG}

We separate the proof into three parts starting with communication, then privacy, and finally the accuracy.

\begin{enumerate}[wide, labelindent=0pt]
\item \textbf{Communication}: Each user $i$ sends the server an answer $a_i$, which is a $2$ dimensional vector. The first coordinate has information about the user's group, \Ie $\mathring{g}\in\cG$, and the second coordinate has information about the user's value, \Ie $\mathring{v}\in\cV$. Therefore, to represent the user's answer, $a_i$, we need $\log(|\cV|)+\log(|\cG|)=1+\log(m)+\log(k)$ bits.

\item \textbf{Privacy}: To prove \eqref{eq:priv_RG}, we consider user $i$ and look at the distribution $\Pr(A_i=a|G_i=g)=\Pr\left(\mathring{G}_i=g',\mathring{V}_i=v|G_i=g\right)$, for all $v\in\cV$ and $g,g'\in\cG$. We separately consider the two cases of $g'=g$ and $g'\neq g$ mirroring the two cases described in Section \ref{sec:Scheme_R}. 
For all $v\in\cV$ and $g,g'\in\cG$ such that $g'= g$,
\begin{multline}
 \Pr\left(\mathring{G}_i=g,\mathring{V}_i=v|G_i=g\right)\\
    =\left((1-\lambda_{vl})p_g(v)+\frac{(1-p_g(v))\lambda_{vl}}{2m-1}\right)(1-\lambda_{gr}),\label{eq:RG_privacy_proof_1}
\end{multline}
which follows from \eqref{eq:RG_def_G_round}, \eqref{eq:RG_def_V_round}, and
\begin{align*}
    \Pr\left(\mathring{V}_i=v|\mathring{G}_i=g,G_i=g\right)\hspace{-1pt}
    =\hspace{-1pt}(1-\lambda_{vl})p_g(v)+\frac{(1-p_g(v))\lambda_{vl}}{2m-1}
\end{align*}
However, for all $v\in\cV$ and $g,g'\in\cG$ such that $g'\neq g$, 
\begin{align} 
\label{eq:RG_privacy_proof_2}
\Pr\left(\mathring{G}_i=g',\mathring{V}_i=v|G_i=g\right)=\frac{\lambda_{gr}}{2m(k-1)},
\end{align}
which follows from \eqref{eq:RG_def_G_round} and \eqref{eq:RG_def_V_round_neq}.

From Definition \ref{def:privacy}, we drop $Q$ from the conditioning because there are no queries assigned to the users in this scheme, also the conditioning on $P$ and $\Theta$ is implicit. Therefore, by substituting \eqref{eq:RG_privacy_proof_1} and \eqref{eq:RG_privacy_proof_2} in \eqref{eq:privacy}, we get 
\begin{multline*}
    e^{\ep_\RG}= \max\bigg\{\rho\max_{g\in\cG,v\in\cV}p_{g}(v)(2m(1-\lambda_{vl})-1)+\lambda_{vl},\\
    \left(\rho\min\limits_{g\in\cG,v\in\cV}p_{g}(v)(2m(1-\lambda_{vl})-1)+\lambda_{vl}\right)^{-1}\bigg\},
\end{multline*}
where $\rho=\frac{2m(k-1)(1-\lambda_{gr})}{\lambda_{gr}}$.

\item \textbf{Accuracy}: 
For all $i\in[n]$, user $i$ sends the server the answer $A_i=\left(\mathring{G}_i,\mathring{V}_i\right)$, where the user's randomized group, $\mathring{G}_i$, is described in \eqref{eq:RG_def_G_round}, and his randomized value, $\mathring{V}_i$, is described in equations \eqref{eq:RG_def_V_round} and \eqref{eq:RG_def_V_round_neq}. 

We define an auxiliary random variable $Z_i$ that functions as an indicator for both user $i$'s randomized group and randomized value. More precisely,  $Z_i$ is a random $k$ dimensional vector (where $k$ is the number of groups), such that $Z_i(j)=V_i$ if $j= \mathring{G}_i$ and $Z_i(j)=0$ otherwise, \Ie $j\neq G_i$.
For all $j\in[k]$ and $v\in\cV$, one readily obtains
\begin{multline}
    \Pr(Z_i(j)=v)=\frac{(1-\theta_g)\lambda_{gr}}{2m(k-1)}
    \\ +\theta_g(1-\lambda_{gr})\left((1-\lambda_{vl})p_j(v)+\frac{(1-p_j(v))\lambda_{vl}}{2m-1}\right).
    \label{eq:prob_Z}
\end{multline}

Since $Z_1,Z_2,\dots,Z_n$ are i.i.d., then following from \eqref{eq:prob_Z} for all $j\in[k]$ and $i\in[n]$ the expectation
\begin{multline}
\label{eq:exp_Z}
    \E\left[\sum_{i\in[n]} Z_i(j)\right]\\
    =\frac{n\theta_g(1-\lambda_{gr})(2m(1-\lambda_{vl})-1)}{2m-1}\E[V_1|G_1=j].
\end{multline}
Then, 
$
    \E\left[\hbS_\RG-\hbS_\QA\right]=0,
$
which follows from \eqref{eq:exp_sum_X_i} and \eqref{eq:exp_Z}. Thus, the estimator of the RG scheme is unbiased.
Moreover,
\small
\begin{multline}
    \E\left[\left( \frac{2m-1}{(1-\lambda_{gr})(2m(1-\lambda_{vl})-1)}\sum_{i=1}^n Z_i(j)\right)^2\right]
    \\=\frac{n(4m^2-1)[2m\theta_g(k-1)(1-\lambda_{gr})\lambda_{vl}
    +\lambda_{gr}(2m-1)(1-\theta_g)]}{6(k-1)(1-\lambda_{gr})^2(2m(1-\lambda_{vl})-1)^2(m+1)^{-1}}\\
    +\frac{n\theta_g(2m-1)}{(1-\lambda_{gr})(2m(1-\lambda_{vl})-1)}\E[V_1^2|G_1=j]\\
    +(n^2-n)\theta_g^2\E[V_1|G_1=j]^2.
    \label{eq:sum_Z_squared}
\end{multline}
\normalsize

Consider the random variables $X_i$, for all $i\in[n]$, described in \eqref{eq:def_X}. For all $i\in[n]$ and $j\in[k]$, notice that the product  $X_i(j)Z_i(j)$ can take on one of these values: 
\begin{align*}
    &X_i(j)Z_i(j)
    \\&=\begin{cases}
    v^2 & \forall v\in\cV \textnormal{, if } X_i(j)=Z_i(j)=v,\\
    vv' & \forall v,v'\in\cV, v'\neq v  \textnormal{, if } X_i(j)=v, \textnormal{ and } Z_i(j)=v',\\
    0 &  \forall v,v'\in\cV  \textnormal{, if } X_i(j)=0 \textnormal{ or } Z_i(j)=0,\\
    \end{cases}
\end{align*}
Then, we can use this to find the expectation
\begin{multline}
\label{eq:exp_X_Z}
    \E\left[X_i(j)Z_i(j)\right]\\=\frac{2m(1-\lambda_{vr})-1}{2m-1}\E\left[V_i^2|G_i=j\right]\theta_j (1-\lambda_{gr}).
\end{multline}

Moreover, since $Z_1,\dots,Z_n$ are i.i.d., $X_1,\dots,X_n$ are i.i.d., and $Z_i$ is independent of $X_\ell$ if $i\neq \ell$, 
\begin{align}
\label{eq:sum_X_Z}
    &\E\left[\sum_{i\in [n]} Z_i(j)\sum_{\ell\in[n]} X_\ell(j)\right]\nonumber\\
    &=\sum_{i\in [n]}\E\left[Z_i(j) X_i(j)\right]+(n^2-n)\hspace{-4pt}\sum_{i,\ell\in [n], i\neq \ell}\hspace{-4pt}\E\left[ Z_i(j)\right]\left[X_\ell(j)\right]\nonumber\\
    &=\frac{(n^2-n)(1-\lambda_{gr})(2m(1-\lambda_{vl})-1)}{2m-1}\E[V_1|G_1=j]^2\theta_g^2\nonumber\\
    &\hspace{60pt}+n(1-\lambda_{gr})(1-\lambda_{vl})\E[V_1^2|G_1=j]\theta_g,
\end{align}
which follows from the substitution of \eqref{eq:exp_sum_X_i}, \eqref{eq:exp_Z}, and \eqref{eq:exp_X_Z}.

Notice that $\hbS_\RG(j)= \frac{2m-1}{(1-\lambda_{gr})(2m(1-\lambda_{vl})-1)}\sum_{i=1}^n Z_i(j)$, and $\bS(j)= \sum_{i=1}^n X_i(j)$. This implies in
\begin{multline}
    \Er_\RG
    =\frac{2m\lambda_{vl}(1-\lambda_{gr})+\lambda_{gr}(2m-1)}{n(1-\lambda_{gr})(2m(1-\lambda_{vl})-1)}\Bigg(\E[V_1^2] \\
+\frac{(4m^2-1)(m+1)}{6(1-\lambda_{gr})(2m(1-\lambda_{vl})-1)}\Bigg),\label{eq:MSE_R_2}
\end{multline}
which follows from substituting \eqref{eq:sum_X_i}, \eqref{eq:sum_Z_squared}, \eqref{eq:sum_X_Z} in \eqref{eq:error_def}, and noting that $\sum_{v\in\cV}v^2=\frac{1}{3}m(m+1)(2m+1)$. 

This proves how we obtained the expression of $\Er_\RG$ as a function of $\lambda_{gr}$ and $\lambda_{vl}$. 
Moreover, one could be interested in the error as a function of a given required privacy $\ep>0$. We give an upper bound of $\Er_\RG$ as a function of $\ep$. We substitute $\lambda_{gr}$ and $\lambda_{vl}$ that minimize the error from Corollary \ref{cor:minimum_RG}, in \eqref{eq:MSE_R_2}.
Then, the error $\Er_\RG$ is  upper bounded by $\cO\left(\frac{m^4k^2}{ne^\ep}\right)$. 
\end{enumerate}

\subsection{Proof of Corollary \ref{cor:minimum_RG}}
\label{app:proof_co_minimum_RG}

We first assume that $p_{\max}> \frac{1}{2m}$ and $p_{\min}< \frac{1}{2m}$, and consider the special case of $p_{\max}=p_{\min}=\frac{1}{2m}$ separately in the end. For ease of notation define 
\begin{align*}
 f(\lambda_{gr},\lambda_{vl}) &=\text{MSE}(\hbS_\RG)\\
 &=\frac{2m\lambda_{vl}(1-\lambda_{gr})+\lambda_{gr}(2m-1)}{(1-\lambda_{gr})(2m(1-\lambda_{vl})-1)}\Bigg(\E[V_1^2] \\
&\hspace{50pt}+\frac{(4m^2-1)(m+1)}{6(1-\lambda_{gr})(2m(1-\lambda_{vl})-1)}\Bigg)n
\end{align*}
which follows directly from \eqref{eq:RG_error}.
To minimize the error of the RG scheme we solve the following optimization problem,
\begin{equation}
\label{eq:proof_opt_1}
\begin{aligned}
& \underset{\lambda_{vl},\lambda_{gr}}{\text{minimize}} 
& & \hspace{-5pt} f(\lambda_{gr},\lambda_{vl}) \\
& \text{subject to}
& & \hspace{-5pt} \textstyle e^\ep\hspace{-2pt}=\hspace{-2pt} \max\hspace{-1pt}\bigg\{\frac{2m(k-1)(1-\lambda_{gr})(p_{\max}(2m(1-\lambda_{vl})-1)+\lambda_{vl})}{\lambda_{gr}(2m-1)},
\\
&&&\hspace{30pt}\textstyle \frac{\lambda_{gr}(2m-1)}{2m(k-1)(1-\lambda_{gr})(p_{\min}(2m(1-\lambda_{vl})-1)+\lambda_{vl})}\bigg\},\\
&&&\hspace{-5pt}\textstyle0\leq\lambda_{vl}<\frac{2m-1}{2m}, 0<\lambda_{gr}<1.
\end{aligned}
\end{equation}
		To solve it, we consider two optimization problems. Consider this first optimization problem, assume its optimal value is attained, and  let $\lambda_{gr}^{(1)}$ and $\lambda_{vl}^{(1)}$ be its optimal points,
\begin{equation}
\label{eq:proof_opt_2}
\begin{aligned}
& \textstyle\underset{\lambda_{vl},\lambda_{gr}}{\text{minimize}} 
& & \textstyle f(\lambda_{gr},\lambda_{vl}) \\
& \textstyle\text{subject to}
& & \textstyle e^\ep= \frac{2m(k-1)(1-\lambda_{gr})(p_{\max}(2m(1-\lambda_{vl})-1)+\lambda_{vl})}{\lambda_{gr}(2m-1)},\\
&&& \textstyle\frac{2m(k-1)(1-\lambda_{gr})(p_{\max}(2m(1-\lambda_{vl})-1)+\lambda_{vl})}{\lambda_{gr}(2m-1)}\\
&&&\hspace{15pt}\textstyle 
\geq
\textstyle\frac{\lambda_{gr}(2m-1)}{2m(k-1)(1-\lambda_{gr})(p_{\min}(2m(1-\lambda_{vl})-1)+\lambda_{vl})},\\
&&&\textstyle 0\leq\lambda_{vl}<\frac{2m-1}{2m}, 0<\lambda_{gr}<1.
\end{aligned}
\end{equation}
Consider this second optimization problem, assume its optimal value is attained, and  let $\lambda_{gr}^{(2)}$ and $\lambda_{vl}^{(2)}$ be its optimal points,
\begin{equation}
\label{eq:proof_opt_3}
\begin{aligned}
& \underset{\lambda_{vl},\lambda_{gr}}{\text{minimize}} 
& & f(\lambda_{gr},\lambda_{vl}) \\
& \text{subject to}
& & \textstyle e^\ep= \frac{\lambda_{gr}(2m-1)}{2m(k-1)(1-\lambda_{gr})(p_{\min}(2m(1-\lambda_{vl})-1)+\lambda_{vl})},\\
&&& \textstyle  \frac{\lambda_{gr}(2m-1)}{2m(k-1)(1-\lambda_{gr})(p_{\min}(2m(1-\lambda_{vl})-1)+\lambda_{vl})}\\
&&&\hspace{15pt}\textstyle 
\geq
\frac{2m(k-1)(1-\lambda_{gr})(p_{\max}(2m(1-\lambda_{vl})-1)+\lambda_{vl})}{\lambda_{gr}(2m-1)},\\
&&&0\leq\lambda_{vl}<1-\frac{1}{2m}, 0<\lambda_{gr}<1.
\end{aligned}
\end{equation}
Then, the solution of \eqref{eq:proof_opt_1} is $\min\left\{f\left(\lambda_{gr}^{(1)},\lambda_{vl}^{(1)}\right),f\left(\lambda_{gr}^{(2)},\lambda_{vl}^{(2)}\right)\right\}$. Therefore, to solve \eqref{eq:proof_opt_1}, we first solve \eqref{eq:proof_opt_2} and \eqref{eq:proof_opt_3}.

\begin{itemize}[wide, labelindent=0pt]
    \item \textit{{Solution of \eqref{eq:proof_opt_2}}:} 
Since $0\leq\lambda_{vl}<\frac{2m-1}{2m}$, we have the following two cases.
\begin{itemize}[leftmargin=+.3in]
    \item[$\circ$] If $e^{2\ep}<\frac{p_{\max}}{p_{\min}}$, from the first condition directly follows 
    \begin{equation}
    \label{eq:proof_opt_4}
    \textstyle
        \lambda_{gr}^{(1)}=\frac{2m(k-1)(p_{\max}(2m(1-\lambda_{vl})-1)+\lambda_{vl})}{2m(k-1)(p_{\max}(2m(1-\lambda_{vl})-1)+\lambda_{vl})+e^\ep}.
    \end{equation}

Since $e^{2\ep}<\frac{p_{\max}}{p_{\min}}$, then $(2m-1)(p_{\max}-p_{\min}e^{2\ep})>0$, and by substituting \eqref{eq:proof_opt_4} in the conditions of \eqref{eq:proof_opt_2}, we get 
    		$$\textstyle 0<\frac{(2m-1)p_{\max}-(2m-1)p_{\min}e^{2\ep}}{(1-2mp_{\min})e^{2\ep}+2mp_{\max}-1}\leq\lambda_{vl}<\frac{1}{2}.$$
    Since $f(\lambda_{gr}^{(1)},\lambda_{vl})$ is increasing in $\lambda_{vl}$, then its minimum is at achieved at the boundary of the domain, \Ie
    \begin{equation*}
    \textstyle
        \lambda_{vl}^{(1)}=\frac{(2m-1)p_{\max}-(2m-1)p_{\min}e^{2\ep}}{(1-2mp_{\min})e^{2\ep}+2mp_{\max}-1}.
    \end{equation*}
    Then substituting this back in \eqref{eq:proof_opt_4}, we get
    \begin{equation*}
    \textstyle
        \lambda_{gr}^{(1)}=\frac{2m(k-1)(p_{\max}-p_{\min})e^\ep}{2m(k-1)(p_{\max}-p_{\min})e^\ep+(1-2p_{\min})e^{2\ep}+2mp_{\max}-1}.
    \end{equation*}
    
    \item[$\circ$] If $e^{2\ep}\geq\frac{p_{\max}}{p_{\min}}$, then $(2m-1)(p_{\max}-p_{\min}e^{2\ep})\leq 0$, and 
   				$$ \textstyle \frac{(2m-1)p_{\max}-(2m-1)p_{\min}e^{2\ep}}{(1-2mp_{\min})e^{2\ep}+2mp_{\max}-1}\leq 0\leq\lambda_{vl}<\frac{1}{2}.$$
    Similarly, since  $f(\lambda_{gr}^{(1)},\lambda_{vl})$ is increasing in $\lambda_{vl}$, then 
    $
        \lambda_{vl}^{(1)}=0.
    $
    And we have that,
    \begin{equation*}
    \textstyle
        \lambda_{gr}^{(1)}=\frac{2m(k-1)p_{\max}}{2m(k-1)p_{\max}+e^\ep},
    \end{equation*}
\end{itemize}

\item \textit{{Solution of \eqref{eq:proof_opt_3}}:} The solution of \eqref{eq:proof_opt_3} follows similarly as that of \eqref{eq:proof_opt_2}, and we have the following two cases.
\begin{itemize}[leftmargin=+.3in]
    \item[$\circ$]If $e^{2\ep}<\frac{p_{\max}}{p_{\min}}$, 
    \begin{equation*}
    \textstyle
        \lambda_{vl}^{(2)}=\frac{(2m-1)p_{\max}-(2m-1)p_{\min}e^{2\ep}}{(1-2mp_{\min})e^{2\ep}+2mp_{\max}-1},
    \end{equation*}
    and
    \begin{equation*}
    \textstyle
        \lambda_{gr}^{(2)}=\frac{2m(k-1)(p_{\max}-p_{\min})e^\ep}{2m(k-1)(p_{\max}-p_{\min})e^\ep+(1-2p_{\min})e^{2\ep}+2mp_{\max}-1}.
    \end{equation*}
    
    \item[$\circ$] If $e^{2\ep}\geq\frac{p_{\max}}{p_{\min}}$, then 
    $
        \lambda_{vl}^{(2)}=0,
    $
    and 
    \begin{equation*}
    \textstyle
        \lambda_{gr}^{(2)}=\frac{2m(k-1)p_{\min}e^\ep}{2m(k-1)p_{\min}e^\ep+1}.
    \end{equation*}
\end{itemize}

\item \textit{Combining the two solutions:} We also have to look at the two cases separately as follows.
    \begin{itemize}[leftmargin=+.3in]
        \item[$\circ$] If $e^{2\ep}<\frac{p_{\max}}{p_{\min}}$, the solution is straightforward, and the optimal points for \eqref{eq:proof_opt_1} are $$\textstyle\lambda_{vl}^*=\frac{(2m-1)p_{\max}-(2m-1)p_{\min}e^{2\ep}}{(1-2mp_{\min})e^{2\ep}+2mp_{\max}-1},$$ and $$\textstyle\lambda_{gr}^*=\frac{2m(k-1)(p_{\max}-p_{\min})e^\ep}{2m(k-1)(p_{\max}-p_{\min})e^\ep+(1-2p_{\min})e^{2\ep}+2mp_{\max}-1}.$$

        \item[$\circ$] If $e^{2\ep}\geq\frac{p_{\max}}{p_{\min}}$, we need to compare $f\left(\lambda_{vl}^{(1)},\lambda_{gr}^{(1)}\right)$ and $f\left(\lambda_{vl}^{(2)},\lambda_{gr}^{(2)}\right)$, and one readily obtains
            \begin{align*}
               \textstyle f\left(\lambda_{vl}^{(1)},\lambda_{gr}^{(1)}\right)\leq f\left(\lambda_{vl}^{(2)},\lambda_{gr}^{(2)}\right).
            \end{align*}
            Therefore, for this case, the optimal points for \eqref{eq:proof_opt_1} are $\lambda_{vl}^*=0$ and $\lambda_{gr}^*=\frac{2m(k-1)p_{\max}}{2m(k-1)p_{\max}+e^\ep}$.
    \end{itemize}
\end{itemize}

This completes the proof for $p_{\max}>\frac{1}{2m}$. If $p_{\max}=p_{\min}=\frac{1}{2m}$, then the second condition of \eqref{eq:proof_opt_1} reduces to $$\textstyle e^\ep=\max\left\{\frac{(k-1)(1-\lambda_{gr})}{\lambda_{gr}},\frac{\lambda_{gr}}{(k-1)(1-\lambda_{gr})}\right\}.$$ 
\Ie $e^\ep$ is not a function of $\lambda_{vl}$. 
Therefore, for this case, the optimal points for \eqref{eq:proof_opt_1} can be readily obtained such that $\lambda_{vl}^*=0$ and $\lambda_{gr}^*=\frac{k-1}{k-1+e^\ep}$.
Combining all the described cases completes the proof.

\begin{remark}
In Corollary \ref{cor:minimum_RG}, we minimize the relative error subject to a fixed privacy parameter $\ep$. Because of the monotonicity of the relative error as a function of $\ep$, an increase in privacy, \Ie smaller $\ep$, cannot decrease the error. Thus, minimizing the error subject to
\begin{multline*}
   \textstyle  e^\ep\leq \max\bigg\{\frac{2m(k-1)(1-\lambda_{gr})(p_{\max}(2m(1-\lambda_{vl})-1)+\lambda_{vl})}{\lambda_{gr}(2m-1)},\\
   \textstyle \frac{\lambda_{gr}(2m-1)}{2m(k-1)(1-\lambda_{gr})(p_{\min}(2m(1-\lambda_{vl})-1)+\lambda_{vl})}\bigg\},
\end{multline*}
is equivalent to solving the optimization \eqref{eq:proof_opt_1}.
\end{remark}


\section{Comparison: Proof of Theorem \ref{thm:QA_vs_RG}}
\label{app:comp_A_C}
We start by sketching the proof of (i) in Theorem \ref{thm:QA_vs_RG}. 
\begin{itemize}
    \item For $k=2$, $m=1$, and $p_1(v)\neq p_2(v') $  or $p_1(v)= p_2(v)=0.5 $ for all $v,v'\in\cV=\{-1,1\}$, we can easily find the exact value of $\lambda$ that satisfies \eqref{eq:priv_C_general}; therefore, we can find the expression for the error of the Q\&A scheme $\Er_\QA(\ep,b)$. Moreover, the minimum error of the RG scheme $\Er_\RG(\ep,b)$ follows from Corollary \ref{cor:minimum_RG}. We find that the limit of the difference of the  errors, $\Er_\RG(\ep,b)-\Er_\QA(\ep,b)$, as $\ep$ goes to zero, is positive. 
    \item For $k>2$ and $m=1$ or $k\geq2$ and $m>1$, the minimum error of the RG scheme $\Er_\RG(\ep,b)$ follows from Corollary \ref{cor:minimum_RG}.
    From Remark \ref{rem:choice_of_rho}, to guarantee a required privacy $\ep$, we can choose any $\lambda\geq\frac{(2m-1)p_{\max}-p_{\min}e^\ep}{2m(p_{\max}-p_{\min}e^\ep)+e^\ep -1}$. We use this $\lambda$ to bound the error of the Q\&A scheme. Finally,  we find that the bound on the limit of the difference of the  errors, $\Er_\RG(\ep,b)-\Er_\QA(\ep,b)$,  as $\ep$ goes to zero, is positive.
\end{itemize}
 
 Now we prove (ii) of Theorem \ref{thm:QA_vs_RG} by showing that there exists an $\ep_\ell>0$, such that for all $\ep>\ep_\ell$, we have $\Er_\QA(\ep,b)>\Er_\RG(\ep,b).$
We first consider the Q\&A scheme. From Remark \ref{rem:choice_of_rho}, there exits an $\ep_0>0$, such that for all $\ep>\ep_0$, the parameter $\lambda=0$ guarantees privacy level $\ep_0$. 
And the error of the Q\&A scheme, as defined in \eqref{eq:QA_ER_ep_B}, for $\lambda=0$, \Ie all $\ep>\ep_0$, is
\begin{align*}
    \Er_\QA(\ep,b)=\frac{1}{6b}\log(2m)(2m+1)(m+1)\left(k-1\right)>0.
\end{align*}
Let $\ep_\ell>\ep_0>\sqrt{\frac{p_{\max}}{p_{\min}}}$, then from Corollary \ref{cor:minimum_RG},  the parameters $\lambda_{vl}^*=0$ and $\lambda_{gr}^*=\frac{2m(k-1)p_{\max}}{2m(k-1)p_{\max}+e^{\ep_\ell}}$ minimize the error of the RG scheme. 
Thus, there exists $\ep_\ell$, such that for all $\ep>\ep_\ell$, 
\begin{align*}
    \Er_\QA(\ep_\ell,b)-\Er_\RG(\ep_\ell,b)&=\frac{\log(2m)(2m+1)(m+1)\left(k-1\right)}{6b}
    \\&>0,
\end{align*}
which completes the proof.

\section*{Acknowledgment}
The authors would like to thank Peter Kairouz for helpful discussions.

\bibliographystyle{IEEEtran}
\IEEEtriggeratref{16}
\bibliography{bibfile}

\end{document}